\documentclass[11pt]{article}

\usepackage{float,graphicx,verbatim,fullpage,hyperref,amssymb,amsmath,amsthm,enumerate,multicol, xspace}
\usepackage[margin=1in]{geometry}
\usepackage{algorithm, algorithmic, cite}

\def\final{1}  
\def\iflong{\iffalse}
\ifnum\final=0  
\newcommand{\knote}[1]{[{\tiny Karthik: \bf #1}]\marginpar{*}}
\newcommand{\jnote}[1]{[{\tiny Jon: \bf #1}]\marginpar{*}}
\newcommand{\junote}[1]{[{\tiny Justin: \bf #1}]\marginpar{*}}
\newcommand{\anote}[1]{[{\tiny Andrew: \bf #1}]\marginpar{*}}
\newcommand{\sidecomment}[1]{}
\else 
\newcommand{\knote}[1]{}
\newcommand{\jnote}[1]{}
\newcommand{\junote}[1]{}
\newcommand{\anote}[1]{}
\newcommand{\sidecomment}[1]{}
\fi  

\newtheorem{theorem}{Theorem}[section]

\newtheorem{lemma}[theorem]{Lemma}

\newtheorem{corollary}[theorem]{Corollary}

\newtheorem{fact}[theorem]{Fact}

\theoremstyle{definition}
\newtheorem{definition}[theorem]{Definition}

\newcommand{\INDSTATE}[1][1]{\STATE\hspace{#1\algorithmicindent}}

\def\R{\mathbb{R}}
\def\E{\mathbb{E}}
\def\N{\mathbb{N}}
\def\Z{\mathbb{Z}}

\newcommand{\cE}{\mathcal{E}}
\newcommand{\cF}{\mathcal{F}}
\newcommand{\cP}{\mathcal{P}}
\newcommand{\cQ}{\mathcal{Q}}
\newcommand{\cR}{\mathcal{R}}
\newcommand{\cS}{\mathcal{S}}
\newcommand\cX{{\{0,1\}^d}}
\newcommand{\cY}{\mathcal{Y}}

\newcommand\poly{\mathrm{poly}}
\newcommand\bits{\{0,1\}}
\newcommand{\pmo}{\{-1,1\}}

\newcommand\card[1]{\left| #1 \right|} 
\newcommand\set[1]{\left\{#1\right\}} 
\newcommand\eps{\varepsilon}
\newcommand\from{\colon}

\newcommand{\idc}{{iterative database construction}\xspace}
\newcommand{\idcs}{{iterative database constructions}\xspace}
\newcommand{\IDC}{{Iterative Database Construction}\xspace}

\newcommand{\dus}{database update sequence\xspace}
\newcommand{\DUS}{Database Update Sequence\xspace}
\newcommand{\update}{\mathbf{U}}

\newcommand{\db}{D}
\newcommand{\dbs}{(\bits^d)^n}
\newcommand{\dbstruct}{\mathcal{P}_{t,W}}

\newcommand{\queryset}{\mathcal{Q}}

\newcommand{\querystep}[1]{y^{(#1)}}
\newcommand{\dbstep}[1]{p^{(#1)}}

\newcommand{\maxupdates}{B}
\newcommand{\updates}{C}
\newcommand{\acc}{\alpha}

\newcommand{\real}{a}
\newcommand{\noisyreal}{\widehat{\real}}

\newcommand{\etal}{et al.}

\newcommand{\noisyrealstep}[1]{\noisyreal^{(#1)}}

\newcommand\san{\mathcal{A}}

\newcommand{\prob}[1]{\Pr\left[#1\right]}

\title{Faster Private Release of Marginals on Small Databases\thanks{Harvard University, School of Engineering and Applied Sciences. Email: \{karthe, jthaler, jullman, atw12\}@seas.harvard.edu. Karthekeyan Chandrasekaran is supported by Simons Fellowship. Justin Thaler is supported by an NSF Graduate Research Fellowship and NSF grants CNS-1011840 and CCF-0915922.  Jonathan Ullman is supported by NSF grant CNS-1237235 and a Siebel Scholarship.  Andrew Wan is supported by NSF grant CCF-0964401 and NSFC grant 61250110218.}}
\author{Karthekeyan Chandrasekaran
\and Justin Thaler 
\and Jonathan Ullman
\and Andrew Wan}
\begin{document}

\maketitle

\begin{abstract}
We study the problem of answering \emph{$k$-way marginal} queries on a database $D \in (\bits^d)^n$, while preserving differential privacy.  The answer to a $k$-way marginal query is the fraction of the database's records $x \in \bits^d$ with a given value in each of a given set of up to $k$ columns.  Marginal queries enable a rich class of statistical analyses on a dataset, and designing efficient algorithms for privately answering marginal queries has been identified as an important open problem in private data analysis.  
For any $k$, we give a differentially private online algorithm that runs in time 
$$
\poly\left(n, \min\left\{ \exp\left(d^{1-\Omega(1/\sqrt{k})}\right), \exp\left(d / \log^{0.99} d\right)  \right\} \right)
$$ 
per query and answers any (adaptively chosen) sequence of $\poly(n)$ $k$-way marginal queries with error at most $\pm 0.01$ on every query, provided $n \gtrsim d^{0.51} $. To the best of our knowledge, this is the first algorithm capable of privately answering marginal queries with a non-trivial worst-case accuracy guarantee for databases containing $\poly(d, k)$ records in time $\exp(o(d))$.  Our algorithm runs the private multiplicative weights algorithm (Hardt and Rothblum, FOCS '10) on a new approximate polynomial representation of the database. 

We derive our representation for the database by approximating the OR function restricted to low Hamming weight inputs using low-degree polynomials with coefficients of bounded $L_1$-norm.  In doing so, we show new upper and lower bounds on the degree of such polynomials, which may be of independent approximation-theoretic interest. First, we construct a polynomial that approximates the $d$-variate $OR$ function on inputs of Hamming weight at most $k$ such that the degree of the polynomial is at most $d^{1-\Omega(1/\sqrt{k})}$ and the $L_1$-norm of its coefficient vector is $d^{0.01}$.  Then we show the following lower bound that exhibits the tightness of our approach: for any $k = o(\log d)$, any polynomial whose coefficient vector has $L_1$-norm $\poly(d)$ that pointwise approximates the $d$-variate OR function on all inputs of Hamming weight at most $k$ must have degree $d^{1-O(1/\sqrt{k})}$. 
\end{abstract}
\newpage

\section{Introduction}
Consider a database $D \in (\bits^d)^n$ in which each of the $n (= |D|)$ rows corresponds to an individual's record, and each record consists of $d$ binary attributes.  The goal of privacy-preserving data analysis is to enable rich statistical analyses on the database while protecting the privacy of the individuals.  In this work, we seek to achieve \emph{differential privacy}~\cite{DworkMcNiSm06}, which guarantees that no individual's data has a significant influence on the information released about the database.

One of the most important classes of statistics on a dataset is its \emph{marginals}.  A marginal query is specified by a set $S \subseteq [d]$ and a pattern $t \in \bits^{|S|}$. The query asks, ``What fraction of the individual records in $D$ has each of the attributes $j \in S$ set to $t_j$?''  A major open problem in privacy-preserving data analysis is to \emph{efficiently} release a differentially private summary of the database that enables analysts to answer each of the $3^{d}$ marginal queries.  A natural subclass of marginals are \emph{$k$-way marginals}, the subset of marginals specified by sets $S \subseteq [d]$ such that $|S| \leq k$.

Privately answering marginal queries is a special case of the more general problem of privately answering \emph{counting queries} on the database, which are queries of the form, ``What fraction of individual records in $D$ satisfy some property $q$?''  Early work in differential privacy~\cite{DinurNi03,BlumDwMcNi05,DworkMcNiSm06} showed how to privately answer any set of counting queries $\cQ$ approximately by perturbing the answers with appropriately calibrated noise, ensuring good accuracy (say, within $\pm .01$ of the true answer) provided $|D| \gtrsim |\cQ|^{1/2}$.

However, in many settings data is difficult or expensive to obtain, and  the requirement that $|D| \gtrsim |\cQ|^{1/2}$ is too restrictive.  For instance, if the query set $\cQ$ includes all $k$-way marginal queries then $|\cQ| \ge d^{\Theta(k)}$, and it may be impractical to collect enough data to ensure $|D| \gtrsim |\cQ|^{1/2}$, even for moderate values of $k$.  Fortunately, a remarkable line of work initiated by Blum et~al.~\cite{BlumLiRo08} and continuing with~\cite{DworkNaReRoVa09,DworkRoVa10,RothRo10,HardtRo10,HardtLiMc12,GuptaRoUl12,JainTh12}, has shown how to privately release approximate answers to any set of counting queries, even when $|\cQ|$ is \emph{exponentially larger} than $|D|$.  For example, the \emph{online private multiplicative weights algorithm} of Hardt and Rothblum~\cite{HardtRo10} gives accurate answers to any (possibly adaptively chosen) sequence of queries $\cQ$ provided $|D| \gtrsim \sqrt{d} \log|\cQ|$.  Hence, if the sequence consists of all $k$-way marginal queries, then the algorithm will give accurate answers provided $|D| \gtrsim k\sqrt{d}$.  Unfortunately, all of these algorithms have running time at least $2^{d}$ per query, even in the simplest setting where $\cQ$ is the set of $2$-way marginals.

Given this state of affairs, it is natural to seek \emph{efficient} algorithms capable of privately releasing approximate answers to marginal queries even when $|D| \ll d^k$.  The most efficient algorithm known for this problem, due to Thaler, Ullman, and Vadhan~\cite{ThalerUlVa12} (building on the work of Hardt, Rothblum, and Servedio~\cite{HardtRoSe12}) 
runs in time $d^{O(\sqrt{k})}$ and releases a summary from which an analyst can compute the answer to any $k$-way marginal query in time $d^{O(\sqrt{k})}$.

Even though $|D|$ can be much smaller than $|\cQ|^{1/2}$, a major drawback of this algorithm and other efficient algorithms for releasing marginals (e.g.~\cite{GuptaHaRoUl11, CheraghchiKlKoLe12, HardtRoSe12, FeldmanKo13, DworkNiTa13}) is that the database still must be significantly larger than $\tilde{\Theta}(k \sqrt{d})$, which we know would suffice for inefficient algorithms.  Recent experimental work of Hardt, Ligett, and McSherry~\cite{HardtLiMc12} demonstrates that for some databases of interest, even the $2^d$-time private multiplicative weights algorithm is practical, and also shows that more efficient algorithms based on adding independent noise do not provide good accuracy for these databases.  Motivated by these findings, we believe that an important approach to designing practical algorithms is to achieve a minimum database size comparable to that of private multiplicative weights, and seek to optimize the running time of the algorithm as much as possible.  In this paper we give the first algorithms for privately answering marginal queries for this parameter regime.  

\subsection{Our Results}

In this paper we give faster algorithms for privately answering marginal queries on databases of size $\widetilde{O}(d^{0.51} / \eps)$, which is nearly the smallest a database can be while admitting any differentially private approximation to marginal queries \cite{BunUlVa13}.

\begin{theorem} \label{thm:main1}
There exists a constant $C > 0$ such that for every $k, d, n \in \N$, $k \leq d$, and every $\eps, \delta > 0$, there is an $(\eps, \delta)$-differentially private online algorithm that, on input a database $D \in (\bits^d)^n$, runs in time 
$$
\poly\left(n, \min\left\{ \exp\left(d^{1-1/C\sqrt{k}}\right), \exp\left(d / \log^{0.99} d\right)  \right\}\right)
$$ 
per query and answers any sequence $\cQ$ of (possibly adaptively chosen) $k$-way marginal queries on $D$ up to an additive error of at most $\pm 0.01$ on every query with probability at least $0.99$, provided that $n \geq C d^{0.51} \log |\cQ| \log(1/\delta) / \eps$.
\end{theorem}

\begin{table} 
\begin{minipage}[center]{\textwidth}
\begin{center}
\begin{tabular}{|c|c|c|}
\hline
Reference & Running Time per Query  & Database Size  \\
\hline
\cite{DinurNi03,DworkNi04,BlumDwMcNi05,DworkMcNiSm06} &$O(1)$ & $\widetilde{O}(d^{k/2})$ \\
\cite{HardtRo10, GuptaRoUl12}& $2^{O(d)}$ & $\widetilde{O}(k\sqrt{d})$  \\
\cite{HardtRoSe12, ThalerUlVa12} & $d^{O(\sqrt{k})}$ & $d^{O(\sqrt{k})}$ \\
\cite{DworkNiTa13} & $d^{O(k)}$ & $\widetilde{O}(d^{\lceil k/2 \rceil / 2})$ \\
This paper & $2^{O(d / \log^{0.99} d)}$ &$ k d^{0.5 + o(1)}$ \\
This paper & $2^{d^{1-\Omega(1/\sqrt{k})}}$ & $\widetilde{O}(kd^{.51})$ \\
\hline
\end{tabular}
\caption{Summary of prior results on differentially private release of $k$-way marginals with error $\pm 0.01$ on every marginal.  Note that the running time ignores dependence on the database size, privacy parameters, and the time required to evaluate the query non-privately.}
\end{center}
\end{minipage}
\end{table}

When $k$ is much smaller than $d$, it may be useful to view our algorithm as an offline algorithm for releasing answers to all $k$-way marginal queries.  This offline algorithm can be obtained simply by requesting answers to each of the $d^{\Theta(k)}$ distinct $k$-way marginal queries from the online mechanism.  In this case we obtain the following corollary.
\begin{corollary} \label{thm:maincor1}
There exists a constant $C > 0$ such that for every $k, d, n \in \N$, $k = O(d / \log d)$, and every $\eps, \delta > 0$, there is an $(\eps, \delta)$-differentially private offline algorithm that, on input a database $D \in (\bits^d)^n$, runs in time 
$$
\poly\left(n, \min\left\{ \exp\left(d^{1-1/C\sqrt{k}}\right), \exp\left(d / \log^{0.99} d\right)  \right\} \right)
$$
and, with probability at least $0.99$, releases answers to every $k$-way marginal query on $D$ up to an additive error of at most $\pm 0.01$, provided that $n \geq C k d^{0.51} \log(1/\delta) / \eps$.
\end{corollary}
Here $\binom{d}{\leq k} := \sum_{i = 0}^{k} \binom{d}{i}$, and the number of $k$-way marginals on $\bits^d$ is bounded by a polynomial in this quantity. See Table 1 for a comparison of relevant results on privately answering marginal queries. 
\vspace{2mm}

\noindent \textbf{Remarks.}
\begin{enumerate}
\item When $k = \Omega(\log^2 d)$, the minimum database size requirement can be improved to $n \geq Ck d^{0.5 + o(1)} \log(1/\delta)/ \eps$, but we have stated the theorems with a weaker bound for simplicity.  (Here $C>0$ is a universal constant and the $o(1)$ is with respect to $d$.)

\item 
Our algorithm can be modified so that instead of releasing approximate answers to each $k$-way marginal explicitly, it releases a summary of the database of size $\widetilde{O}(kd^{0.01})$ from which an analyst can compute an approximate answer to any $k$-way marginal in time $\widetilde{O}(kd^{1.01})$.
\end{enumerate}
\medskip
A key ingredient in our algorithm is a new approximate representation of the database using polynomial approximations to the $d$-variate OR function restricted to inputs of Hamming weight at most $k$. For any such polynomial, the degree determines the runtime of our algorithm, while the $L_1$-weight of the coefficient vector determines the minimum required database size. Although low-degree low $L_1$-weight polynomial approximations to the OR function have been studied in the context
of approximation theory and learning theory \cite{ThalerCOLT12}, our setting requires an approximation only over a restricted subset of the inputs.
When the polynomial needs to approximate the OR function only on a subset of the inputs, is it possible to reduce the degree and $L1$-weight (in comparison to \cite{ThalerCOLT12}) of the polynomial?   

Our main technical contribution addresses this variant of the polynomial approximation problem.
We believe that our construction of such polynomials (Theorem \ref{thm:poly-construction}) as well as the lower bound (Theorem \ref{thm:main2}) could be of independent approximation-theoretic interest.
The following theorem shows a construction of polynomials that achieve better degree and $L_1$-weight in comparison to \cite{ThalerCOLT12} for small values of $k$.
Let $\text{OR}_d: \{-1, 1\}^d \rightarrow \{-1, 1\}$ denote the OR function on $d$ variables with the convention that $-1$ is TRUE, and for any vector $x \in \{-1, 1\}^d$, let $|x|$ denote the number of coordinates of $x$ equal to $-1$.


\begin{theorem}
\label{thm:poly-construction}
Let $k\in [d]$. For some constant $C>0$, there exists a polynomial $p$ such that 
\begin{enumerate}[(i)]
\item $|p(x)-OR_d(x)|\le 1/400$ for every $x\in \{-1,1\}^d:|x|\le k$,
\item the $L_1$-weight of the coefficient vector of $p$ is at most $d^{0.01}$, and
\item the degree of $p$ is at most 
\[
\min\left\{d^{1-\frac{1}{C\sqrt{k}}},\frac{d}{\log^{0.995}{d}}\right\}.
\]
\end{enumerate}
\end{theorem}

The degree bound of $d/\log^{0.995}{d}$ in the above theorem follows directly from techniques developed in \cite{ThalerCOLT12},
while the degree bound of $d^{1-\frac{1}{C\sqrt{k}}}$
requires additional insight.
We also show a lower bound to exhibit the tightness of our construction.


\begin{theorem}
\label{thm:main2}
Let $k = o(\log d)$, and let $p$ be a real $d$-variate polynomial satisfying $|p(x) - \text{OR}_d(x)| \leq 1/6$ for all $x \in \{-1, 1\}^d$ with $|x| \leq k$. If the $L_1$-weight of the coefficient vector of $p$ is $d^{O(1)}$, then the degree of $p$ is at least 
$d^{1-O(1/\sqrt{k})}$.
\end{theorem}

We note that our algorithmic approach for designing efficient private data release algorithms would work equally well if we have any small set of functions whose low-$L_1$-weight linear combinations approximate disjunctions restricted to inputs of Hamming weight at most $k$. Our lower bound limits the applicability of our approach if we choose to use low-degree monomials as the set of functions. We observe that this also rules out several natural candidates that can themselves be computed exactly by a low-weight polynomial of low-degree (e.g., the set of small-width conjunctions).
There is some additional evidence from prior work that 
low-degree monomials may be the optimal choice: 
if we only care about the \emph{size}
of the set of functions used to approximate
disjunctions on inputs of Hamming weight at most $k$,
then prior work shows that low-degree monomials
are indeed optimal \cite{sherstov-pm} (see also Section 5 in the
full version of \cite{ThalerUlVa12}).
It remains an interesting open
question to determine whether this optimality still holds 
when we restrict the $L_1$ weight of the linear combinations
used in the approximations to be $\poly(d)$.




\subsection{Techniques}
For notational convenience, we focus on \emph{monotone $k$-way disjunction queries}. However, our results extend straightforwardly to general non-monotone $k$-way marginal queries via simple transformations on the database and queries.  A monotone $k$-way disjunction is specified by a set $S \subseteq [d]$ of size $k$ and asks what fraction of records in $D$ have at least one of the attributes in $S$ set to $1$.

Following the approach introduced by Gupta et al.~\cite{GuptaHaRoUl11} and developed into a general theory in~\cite{HardtRoSe12}, we view the problem of releasing answers to conjunction queries as a learning problem.  That is, we view the database as specifying a function $f_D \from \pmo^d \to [0,1]$, in which each input vector $s \in \pmo^d$ is interpreted as the indicator vector of a set $S \subseteq \set{1,\dots,d}$, with $s_i = -1$ iff $i \in S$, and $f_D(s)$ equals the evaluation of the conjunction query specified by $S$ on the database $D$.  Then, our goal is to privately learn to approximate the function $f_{D}$; this is accomplished in \cite{HardtRoSe12} by approximating $f_D$ succinctly with polynomials and learning the polynomial privately. Polynomial approximation is central to our approach as well, as we explain below.


We begin with a description of how the parameters of the online learning algorithm determine the parameters of the online differentially private learning algorithm. We consider the ``IDC framework'' \cite{GuptaRoUl12}---which captures the private multiplicative weights algorithm \cite{HardtRo10} among others~\cite{RothRo10, GuptaRoUl12, JainTh12}---for deriving differentially private online algorithms from any \emph{online learning algorithm} that may not necessarily be privacy preserving.

Informally, an online learning algorithm is one that takes a (possibly adaptively chosen) sequence of inputs $s_1,s_2,\dots$ and returns answers $a_1,a_2,\dots$ to each, representing ``guesses'' about the values $f_D(s_1), f_D(s_2),\dots$ for the unknown function $f_D$.  After making each guess $a_i$, the learner is given some information about the value of $f_D(s_i)$.  The quantities of interest are the running time required by the online learner to produce each guess $a_i$ and the number of ``mistakes'' made by the learner, which is the number of rounds $i$ in which $a_i$ is ``far'' from $f_D(s_i)$. Ultimately, for the differentially private algorithm derived in the IDC framework, the notion of far will correspond to the accuracy, the per query running time will essentially be equal to the running time of the online learning algorithm, and the minimum database size required by the private algorithm will be proportional to the square root of the number of mistakes. 

We next describe the well-known technique of deriving faster online learning algorithms that commit fewer mistakes using \emph{polynomial approximations} to the target function.
Indeed, it is well-known that if $f_D$ can be approximated to high accuracy 
by a $d$-variate polynomial $p_D \from \pmo^d \to \R$ of degree $t$ and $L_1$-weight at most $W$, where the weight is defined to be the sum of the absolute values of the coefficients, then there is an online learning algorithm that runs in time $\poly\big(\binom{d}{\leq t}\big)$ and makes $O(W^2d)$ mistakes.  Thus, if $t \ll d$, the running time of such an online learning algorithm will be significantly less than $2^d$ and the number of mistakes (and thus the minimum database size of the resulting private algorithm) will only blow up by a factor of $W$. 

Consequently, our goal boils down to constructing the best possible polynomial representation $p_D$ for any database $D$ -- one with low-degree, low-$L_1$-weight such that $|p_D(s) - f_D(s)|$ is small for all vectors $s \in \pmo^d$ corresponding to monotone $k$-way disjunction queries. 
To accomplish this goal, it is sufficient to construct a low-degree, low-$L_1$-weight polynomial that can approximate the $d$-variate OR function on inputs of Hamming weight at most $k$ (i.e., those that have $-1$ in at most $k$ indices).
Such problems are well-studied in the approximation-theory literature, however our variant requires polynomials to be accurate only on a restricted subset of inputs. 
In fact, the existence of a polynomial with degree $d / \log^{0.99} d$ and $L_1$-weight $d^{o(1)}$ that approximates the $d$-variate OR function on all inputs follows from the work of Servedio et al.~\cite{ThalerCOLT12}. 
We improve these bounds for small values of $k$ by constructing an approximating polynomial that has degree $d^{1-\Omega(1/\sqrt{k})}$ and $L_1$-weight $d^{0.01}$. 



We also prove a new approximation-theoretic lower bound for polynomials that seek to approximate a target function for a restricted subset of inputs. 
Specifically, we show that for any $k = o(\log d)$, any polynomial $p$ of weight $\poly(d)$ that satisfies $|p(s) - \text{OR}(s)| \leq 1/6$ for all inputs $s \in \pmo^d$ of Hamming weight at most $k$ must have degree $d^{1-O(1/\sqrt{k})}$.
We prove our lower bound by expressing the problem of constructing such a low-weight, low-degree polynomial $p$ as a linear program, and exhibiting an explicit solution to the dual of this linear program. 
Our proof is inspired by recent work of Sherstov \cite{sherstov-intersecths, sherstov-pm, sherstov-direct} and Bun-Thaler \cite{bunthaler}. 

\subsection{Related Work}

\paragraph{Other Results on Privately Releasing Marginals.}
In work subsequent to our result, Dwork et al.~\cite{DworkNiTa13} show how to privately release marginals in a very different parameter regime.  Their algorithm is faster than ours, running in time $\poly(\binom{d}{\leq k})$, and has better dependence on the error parameter.  However, their algorithm requires that the database size is $\widetilde{\Omega}(d^{\lceil k/2 \rceil /2})$ for answering with error $\pm 0.01$.  This size is comparable to the optimal $\widetilde{\Omega}(k\sqrt{d})$ only when $k\leq 2$.  In contrast, our algorithm has nearly-optimal minimum database size for every choice of $k$.

While we have focused on accurately answering \emph{every} $k$-way marginal query, or more generally every query in a sequence of marginal queries, several other works have considered more relaxed notions of accuracy.  These works show how to efficiently release a summary of the database from which an analyst can efficiently compute an approximate answer to marginal queries, with the guarantee that the \emph{average} error of a marginal query is at most $.01$, when the query is chosen from a particular distribution.  In particular, Feldman and Kothari~\cite{FeldmanKo13} achieve small average error over the uniform distribution with running time and database size $\tilde{O}(d^2)$; Gupta et al.~\cite{GuptaHaRoUl11} achieve small average error over any product distribution with running time and minimum database size $\poly(d)$; finally Hardt et al.~\cite{HardtRoSe12} show how to achieve small average error over arbitrary distributions with running time and minimum database size $2^{\tilde{O}(d^{1/3})}$.  All of these results are based on the approach of learning the function $f_D$.

Several works have also considered information theoretic bounds on the minimum database size required to answer $k$-way marginals.  Kasiviswanathan et al.~\cite{KasiviswanathanRuSmUl10} showed that $|D| \geq \min\{1/\alpha^2, d^{k/2}/\alpha\}$ is necessary to answer all $k$-way marginals with error $\pm \alpha$.  De~\cite{De12} extended this result to hold even when accuracy $\pm \alpha$ can be violated for a constant fraction of $k$-way marginals.  In our regime, where $\alpha = \Omega(1)$, their results do not give a non-trivial lower bound.  In forthcoming work, Bun, Ullman, and Vadhan~\cite{BunUlVa13} have proven a lower bound of $|D| \geq \tilde{\Omega}(k\sqrt{d})$, which is nearly optimal for $\alpha = \Omega(1)$.

\paragraph{Relationship with Hardness Results for Differential Privacy.}
Ullman~\cite{Ullman13} (building on the results of Dwork et al.~\cite{DworkNaReRoVa09}), showed that any $2^{o(d)}$-time differentially private algorithm that answers \emph{arbitrary} counting queries can only give accurate answers if $|D| \gtrsim |\cQ|^{1/2}$, assuming the existence of exponentially hard one-way functions.  Our algorithms have running time $2^{o(d)}$ and are accurate when $|D| \ll |\cQ|^{1/2}$, and thus show a separation between answering marginal queries and answering arbitrary counting queries.

When viewed as an offline algorithm for answering all $k$-way marginals, our algorithm will return a list of values containing answers to each $k$-way marginal query.  It would in some cases be more attractive if we could return a \emph{synthetic database}, which is a new database $\widehat{D} \in (\bits^d)^{\widehat{n}}$ whose rows are ``fake'', but such that $\widehat{D}$ approximately preserves many of the statistical properties of the database $D$ (e.g., all the marginals).  Some of the previous work on counting query release has provided synthetic data ~\cite{BarakChDwKaMcTa07,BlumLiRo08,DworkNaReRoVa09,DworkRoVa10,HardtLiMc12}.

Unfortunately, Ullman and Vadhan~\cite{UllmanVa11} (building on~\cite{DworkNaReRoVa09}) have shown that no differentially private sanitizer with running time $\poly(d)$ can take a database $D \in (\bits^d)^n$ and output a private synthetic database $\widehat{D}$, all of whose $2$-way marginals are approximately equal to those of $D$, assuming the existence of one-way functions.  They also showed that under certain strong cryptographic assumptions, there is no differentially private sanitizer with running time $2^{d^{1-\Omega(1)}}$ can output a private synthetic database, all of whose $2$-way marginals are approximately equal to those of $D$.  Our algorithms indeed achieve this running time and accuracy guarantee when releasing $k$-way marginals for constant $k$, and thus it may be inherent that our algorithms do not generate synthetic data.

\paragraph{Relationship with Results in Approximation Theory.}
Servedio et al.~\cite{ThalerCOLT12} focused on developing low-weight, low-degree polynomial threshold functions (PTFs) for decision lists, motivated by applications in computational learning theory.
As an intermediate step in their PTF constructions, they constructed low-$L_1$-weight, low-degree polynomials that approximate the OR function on \emph{all} Boolean inputs.
Our construction of lower-weight, lower-degree polynomials that approximate the OR function on low Hamming weight inputs is inspired by and builds on Servedio et al.'s construction
of approximations that are accurate on all Boolean inputs.

The proof of our lower bound is inspired by recent work that has established new approximate degree lower bounds via the construction of dual solutions to certain linear programs.
In particular, Sherstov \cite{sherstov-intersecths} showed that approximate degree and PTF degree behave roughly multiplicatively under function composition, while Bun and Thaler \cite{bunthaler} 
gave a refinement of Sherstov's method in order to resolve the approximate degree of the two-level AND-OR tree, and also gave an explicit dual witness for the approximate degree of any symmetric Boolean function. We extend these lower bounds along two directions: (1) we show degree lower bounds that take into account the $L_1$-weight of the coefficient vector of the approximating polynomial, and
(2) our lower bounds hold even when we only require the approximation to be accurate on inputs of low Hamming weight, while prior work only considered approximations that are accurate on all Boolean inputs.

Some prior work has studied the degree of polynomials that point-wise approximate \emph{partial} Boolean functions \cite{sherstov-direct, sherstov-multi}. Here, a function $f: Y \rightarrow \mathbb{R}$ is said to be partial if its domain $Y$ is a strict subset of $\{-1, 1\}^d$, and a polynomial $p$ is said to $\epsilon$-approximate $f$ if 

\begin{enumerate}
\item $|f(x) - p(x)| \leq \epsilon$ for all $x \in Y$, and 
\item $|p(x)| \leq 1+\epsilon$ for all $x \in \{-1, 1\}^d \setminus Y$.
\end{enumerate}
In contrast, our lower bounds apply even in the absence of Condition 2, i.e., when $p(x)$ is allowed to take arbitrary values on inputs in $ \{-1, 1\}^d \setminus Y$.




Finally, while our motivation is private data release, our approximation theoretic results are similar in spirit to recent work of Long and Servedio \cite{LS13}, who are motivated by applications in computational learning theory. Long and Servedio consider halfspaces $h$ defined on inputs of small Hamming weight, and (using different techniques very different from ours) 
give upper and lower bounds on the weight of these halfspaces when represented as linear threshold functions. \\

\noindent{\bf Organization. }
In Section \ref{sec:privacy}, we describe our private online algorithm and show that it yields the claimed accuracy given the existence of sufficiently low-weight polynomials that approximate the $d$-variate OR function on inputs of low Hamming weight.  The results of this section are a combination of known techniques in differential privacy~\cite{RothRo10, HardtRo10, GuptaRoUl12} and learning theory (see e.g.,~\cite{KlivansSe04}).  Readers familiar with these literatures may prefer to skip Section~\ref{sec:privacy} on first reading.  In Section \ref{sec:upperbound}, we give our polynomial approximations to the OR function, both on low Hamming weight Boolean inputs and on all Boolean inputs.  Finally, in Section \ref{sec:lowerbounds}, we state and prove our lower bounds for polynomial approximations to the OR function on restricted inputs.

\section{Preliminaries}
\label{sec:prelims}

\subsection{Differentially Private Sanitizers}\label{sec:sans}
Let a \emph{database} $D \in (\cX)^n$ be a collection of $n$ rows $x^{(1)}, \dots, x^{(n)}$ from a \emph{data universe} $\cX$.  We say that two databases $D,D' \in (\cX)^n$ are \emph{adjacent} if they differ only on a single row, and we denote this by $D \sim D'$.

Let $\san: (\cX)^n \to \cR$ be an algorithm that takes a database as input and outputs some data structure in $\cR$.  We are interested in algorithms that satisfy \emph{differential privacy}.
\begin{definition}[Differential Privacy~\cite{DworkMcNiSm06}]\label{def:dp} An algorithm $\san\from (\cX)^n \to \cR$ is \emph{$(\eps, \delta)$-differentially private} if for every two adjacent databases $D \sim D' \in (\cX)^n$ and every subset $S \subseteq \cR$,
$$
\prob{\san(D) \in S} \leq e^{\eps} \prob{\san(D') \in S} + \delta.
$$
\end{definition}

Since a sanitizer that always outputs $\bot$ satisfies Definition~\ref{def:dp}, we focus on sanitizers that are accurate.  In particular, we are interested in sanitizers that give accurate answers to \emph{counting queries}.  A counting query is defined by a boolean predicate $q\from \cX \to \bits$.  Abusing notation, we define the evaluation of the query $q$ on a database $D \in (\cX)^n$ to be
$
q(D) = \frac{1}{n} \sum_{i=1}^{n} q(x^{(i)}).
$
Note that the value of a counting query is in $[0,1]$.  We use $\cQ$ to denote a set of counting queries.

For the purposes of this work, we assume that the range of $\san$ is simply $\R^{|\cQ|}$.  That is, $\san$ outputs a list of real numbers representing answers to each of the specified queries.  

\begin{definition}[Accuracy]\label{def:acc}
The output of $\san(D)$, $a = (a_q)_{q \in \cQ}$, is \emph{$\alpha$-accurate} for the query set $\cQ$ if
$$
\forall q \in \cQ,\; |a_q - q(D)| \leq \alpha
$$
A sanitizer is \emph{$(\alpha, \beta)$-accurate} for the query set $\cQ$ if for every database $D$, $\san(D)$ outputs $a$ such that with probability
at least $1-\beta$, the output $a$ is $\alpha$-accurate for $\cQ$,
where the probability is taken over the coins of $\san$.
\end{definition}

We remark that the definition of both differential privacy and $(\alpha, \beta)$-accuracy extend straightforwardly to the online setting.  Here the algorithm receives a sequence of $\ell$ (possibly adaptively chosen) queries from $\cQ$ and must give an answer to each before seeing the rest of the sequence.  Here we require that with probability at least $1-\beta$, \emph{every} answer output by the algorithm is within $\pm \alpha$ of the true answer on $D$.  See e.g., ~\cite{HardtRo10} for an elaborate treatment of the online setting.

\subsection{Query Function Families} \label{sec:Qfamilies}
Given a set of queries of interest, $\cQ$ (e.g., all marginal queries), we think of the database $D$ as specifying a function $f_D$ mapping queries $q$ to their answers $q(D)$.
We now describe this transformation more formally:

\begin{definition}[$\cQ$-Function Family] \label{def:Qfunction}
Let $\cQ = \set{q_y}_{y \in Y_{\cQ} \subseteq \pmo^m}$ be a set of counting queries on a data universe $\cX$, where each query is indexed by an $m$-bit string.  We define the \emph{index set of $\cQ$} to be the set $Y_{\cQ} = \set{y \in \pmo^m \mid q_y \in \cQ}$. 

We define the \emph{$\cQ$-function family} $\cF_{\cQ} = \set{ f_{x}: \pmo^m \to [0,1]}_{x \in \cX}$ as follows:
For every possible database row $x \in \cX$, the function $f_{x}: \pmo^m \to [0,1]$ is defined as $f_{x}(y) = q_y(x)$.  
Given a database $D \in (\cX)^n$ we define the function $f_{\cQ,D}: \pmo^m \to [0,1]$ where $f_{\cQ,D}(y) = \frac{1}{n} \sum_{i=1}^{n} f_{x^{(i)}}(y)$.  When $\cQ$ is clear from context we will drop the subscript $\cQ$ and simply write $f_x$, $f_D$, and $\cF$. 
\end{definition}

When $\cQ$ is the set of all monotone $k$-way disjunctions on a database $D \in (\bits^d)^n$, the queries are defined by sets $S \subseteq [d]$ , $|S| \leq k$.  In this case, we represent each query by the $d$-bit $-1/1$ indicator vector $y_S$ of the set $S$, where $y_S(i)=-1$ if and only if $i\in S$. Thus, $y_S$ has at most $k$ entries that are $-1$.  Hence, we can take $m = d$ and $Y_{\cQ} = \set{ y \in \pmo^d \mid \sum_{j=1}^{d} \mathbf{1}_{\set{y_i = -1}} \leq k}$. 

\subsection{Low-Weight Polynomial Approximations} \label{sec:lwapprox}

Given an $m$-variate real polynomial $p \from \pmo^m \to \R$,
\[
p(y) = \sum_{S \subseteq [m] } c_S \cdot \prod_{i \in S} y_i,
\]
we define the degree, weight $w(\cdot)$ and {\em non-constant weight} $w^*(\cdot)$ of the polynomial as follows:
\begin{align*}
\text{deg}(p)&:=\max\{|S|:S\subseteq [m], c_S\neq 0\},\\
w(p)&:=\sum_{S\subseteq [m]}|c_S|, and\\
w^*(p)&:=\sum_{S\subseteq [m],S\neq \emptyset}|c_S|.
\end{align*}

We use $\binom{[m]}{\leq t}$ to denote $\set{S \subseteq [m] \mid |S| \leq t}$ and $\binom{m}{\leq t} = \card{\binom{[m]}{\leq t}} = \sum_{j = 0}^{t} \binom{m}{j}$.  

We will attempt to approximate the functions $f_{x}: \pmo^m \to \bits$ on all the indices in $Y_{\cQ}$ by a family of polynomials with low degree and low weight. Formally and more generally: 

\begin{definition}[Restricted Approximation by Polynomials] \label{def:approxbypolys}
Given a function $f\from Y\rightarrow \R$, where $Y\subseteq \R^m$, and a subset $Y'\subseteq Y$, we denote the restriction of $f$ to $Y'$ by $f|_{Y'}$. 
Given an $m$-variate real polynomial $p$, we say that $p$ is a {\em $\gamma$-approximation} to the restriction $f|_{Y'}$, if $|f(y)-p(y)|\le \gamma$ $\forall y\in Y'$.
Notice there is no restriction whatsoever placed on $p(y)$ for $y \in Y \setminus Y'$.

Given a family of $m$-variate functions $\cF = \set{f_x\from Y\to\R}_{x \in \cX}$, where $Y\subseteq \R^m$, a set $Y' \subseteq Y$ we use $\cF|_{Y'}=\set{f_x|_{Y'}}_{x \in \cX}$ to denote the family of restricted functions.
Given a family $\cP$ of $m$-variate real polynomials, we say that the family $\cP$ is a \emph{$\gamma$-approximation to the family $\cF|_{Y'}$} if for every $x \in \cX$, there exists $p_{x} \in \cP$ that is a $\gamma$-approximation to $f_x|_{Y'}$.
\end{definition}

Let $H_{m,k}=\{x\in \pmo^m: \sum_{i=1}^m(1-x_i)/2\le k\}$ denote the set of inputs of Hamming weight at most $k$. We view the $d$ variate OR function, $\text{OR}_d$ as mapping inputs from $\pmo^d$ to $\pmo$, with the convention that $-1$ is TRUE and $1$ is FALSE. Let $\cP_{t,W}(m)$ denote the family of all $m$-variate real polynomials of degree $t$ and weight $W$. For the upper bound, we will show that for certain small values of $t$ and $W$, the family $\cP_{t,W}(d)$ is a $\gamma$-approximation to the family of all disjunctions restricted to $H_{d,k}$.

\begin{fact} \label{fact:disj2}
If $\cQ$ is the set of all monotone $k$-way disjunctions on a database $D \in (\bits^d)^n$, $\cF$ is its function family, and $Y = H_{d,k}$ is its index set, then $\cP_{t,W}(d)$ is a $\gamma$-approximation to the restriction $\cF|_{Y}$ if and only if there is a degree $t$ polynomial of weight $O(W)$ that $\gamma$-approximates $OR_{d}|_{H_{d,k}}.$
\end{fact}

The fact follows easily by observing that for any $x \in \bits^d$, $y \in \pmo^d$, $$f_x(y) = \bigvee_{i \in x} \mathbf{1}_{\set{y_i = -1}} = \frac{1-\text{OR}_d(y_1^{x_1}, \dots, y_d^{x_d})}{2}.$$

For the lower bound, we will show that any collection of polynomials with small weight that is a $\gamma$-approximation to the family of disjunctions restricted to $H_{m,k}$ should have large degree. We need the following definitions:
\begin{definition}[Approximate Degree]\label{def:approx-degree}
Given a function $f\from Y\rightarrow \R$, where $Y\subseteq \R^m$, the $\gamma$-approximate degree of $f$ is
\[
\text{deg}_{\gamma}(f):=\text{min}\{d:\exists\text{ real polynomial $p$ that is a $\gamma$-approximation to $f$, $\text{deg}(p)=d$}\}.
\]
Analogously, the $(\gamma,W)$-approximate degree of $f$ is
\[
\text{deg}_{(\gamma,W)}(f):=\text{min}\{d:\exists\text{ real polynomial $p$ that is a $\gamma$-approximation to $f$, $\text{deg}(p)=d$, $w(p)\le W$}\}.
\]
It is clear that $\text{deg}_{\gamma}(f)=\text{deg}_{(\gamma,\infty)}(f)$. 

We let $w^*(f,t)$ denote the \emph{degree-$t$ non-constant margin weight} of $f$, defined to be:
\[
w^*(f,t):=\min\{w^*(p):\exists \text{ real polynomial $p$ s.t. } \text{deg}(p)\le t, f(y)p(y)\ge 1\ \forall\ y\in Y\}.
\]
\end{definition}
The above definitions extend naturally to the restricted function $f|_{Y'}$.

\vspace{2mm}
Our definition of non-constant margin weight is closely related to the well-studied notion of the degree-$t$ polynomial threshold function (PTF) weight of $f$ (see e.g., \cite{sherstov-pm}),
which is defined as $\min_p w(p)$, where the minimum is taken over all degree-$t$ polynomials $p$ with integer coefficients, such that $f(x) = \text{sign}(p(x))$ for all $x \in \{-1, 1\}^d$. 
Often, when studying PTF weight, the requirement that $p$ have integer coefficients is used only to ensure that $p$ has non-trivial margin, i.e. that $|p(x)| \geq 1$ for all $x \in \{-1, 1\}^d$; this is precisely 
the requirement captured in our definition of non-constant margin weight. We choose to work with margin weight because
it is a cleaner quantity to analyze using linear programming duality; PTF weight can also be studied 
using LP duality, but  the integrality constraints on the coefficients of $p$ introduces an integrality gap that causes some loss in the analysis (see e.g., Sherstov \cite[Theorem 3.4]{sherstov-pm} and Klauck \cite[Section 4.3]{klauck}).



\section{From Low-Weight Approximations to Private Data Release}
\label{sec:privacy}
In this section we show that low-weight polynomial approximations imply data release algorithms that provide approximate answers even on small databases.  
The main goal of this section is to prove the following theorem. 

\begin{theorem}\label{thm:poly-approx-gives-DP-algorithm}
Given $\alpha,\beta,\eps,\delta > 0$, and a family of linear queries $\cQ$ with index set $Y\subseteq \{-1,1\}^m$. 
Suppose for some $t \leq m$, $W>0$, the family of polynomials $\cP_{t,W}(m)$ $(\alpha/4)$-approximates the function family $\cF_{\cQ}|_{Y}$. Then there exists an $(\epsilon,\delta)$-differentially private online algorithm that is $(4\alpha,\beta)$-accurate for any sequence of $\ell$ (possibly adaptively chosen) queries from $\cQ$ on a database $D\in (\cX)^n$, provided
\[
n\ge \frac{128 W \log{(\ell/\beta)}\log{(4/\delta)}}{\alpha^2 \eps}\sqrt{\log{\left(2\binom{m}{\le t}+1\right)}}.
\]
The private algorithm has running time $\poly\left(n, \binom{m}{\le t},\log{W},\log{(1/\alpha),\log(1/\beta),\log(1/\eps),\log(1/\delta)}\right)$.
\end{theorem}

We note that the theorem can be assembled from known techniques in the design and analysis of differentially private algorithms and online learning algorithms. We include the proof of the theorem here for the sake of completeness, as to our knowledge they do not explicitly appear in the privacy literature.

We construct and analyze the algorithm in two steps.  First, we use standard arguments to show that the non-private multiplicative weights algorithm can be used to construct a suitable online learning algorithm for $f_{\cQ, D}$ whenever $f_{\cQ, D}$ can be approximated by a low-weight, low-degree polynomial.  Here, a suitable online learning algorithm is one that fits into the IDC framework of Gupta et al.~\cite{GuptaRoUl12}.  We then apply the generic conversion from IDCs to differentially private online algorithms~\cite{RothRo10, HardtRo10, GuptaRoUl12} to obtain our algorithm.

\subsection{IDCs}
We start by providing the relevant background on the \emph{iterative database construction} framework.  An IDC will maintain a sequence of functions $f^{(1)}_{\db}, f^{(2)}_{\db}, \dots$ that give increasingly good approximations to the $f_{\cQ, \db}$.  In our case, these functions will be low-degree polynomials.  Moreover, the mechanism produces the next approximation in the sequence by considering only one query $y^{(t)}$ that ``distinguishes'' the real database in the sense that $|f^{(t)}(y^{(t)}) - f_D(y^{(t)})|$ is large.

\begin{definition}[IDC~\cite{RothRo10,HardtRo10,GuptaRoUl12}]
Let $\cQ = \{q_{y}\}_{y \in Y_{\cQ} \subseteq \{-1,1\}^m}$ be a family of counting queries indexed by $m$-bit strings.  Let $\update$ be an algorithm mapping a function $f \from Y_{\cQ} \to \R$, a query $y \in Y_{\cQ}$, and a real number $a$ to a new function $f'$.  Let $\db \in (\bits^d)^n$ be a database and $\alpha > 0$ be a parameter.  Consider the following game with an adversary.  Let $f^{(1)}_{\db}$ be some function.  In each round $t = 1,2,\dots$:
\begin{enumerate}
\item The adversary chooses a query $y^{(t)} \in Y_{\cQ}$ (possibly depending on $f^{(t)}_{\db}$).
\item If $|f^{(t)}_{\db}(y^{(t)}) - f_{\db}(y^{(t)})| > \alpha$, then we say that the algorithm has \emph{made a mistake}.
\item If the algorithm made a mistake, then it receives a value $a^{(t)} \in \R$ such that $|f^{(t)}_{\db}(y^{(t)}) - f_{\db}(y^{(t)})| \leq \alpha/2$ and computes a new function $f^{(t+1)}_{\db} = \update(f^{(t)}_{\db}, y^{(t)}, a^{(t)})$.  Otherwise let $f^{(t+1)}_{\db} = f^{(t)}_{\db}$.
\end{enumerate}

If the number of rounds $t$ in which the algorithm makes a mistake is at most $\maxupdates$ for every adversary, then $\update$ is a \emph{iterative database construction for $\cQ$ with mistake bound $\maxupdates$}.
\end{definition}

\begin{theorem}[Variant of \cite{GuptaRoUl12}] \label{thm:winnow} For any $\alpha > 0$, and any family of queries $\cQ$, if there is an iterative database construction for $\cQ$ with mistake bound $\maxupdates$, then there is an $(\eps, \delta)$-differentially private online algorithm that is $(4\alpha, \beta)$-accurate for any sequence of $\ell$ (possibly adaptively chosen) queries from $\cQ$ on a database $\db \in (\bits^d)^n$, so long as
$$
n \geq \frac{32 \sqrt{\maxupdates} \log (\ell/\beta) \log(4/\delta)}{\alpha \eps}.
$$
Moreover, if the iterative database construction, $\update$, runs in time $T_{\update}$, then the private algorithm has running time $\poly(n, T_{\update},\log(1/\alpha), \log(1/\beta), \log(1/\eps),\log(1/\delta))$ per query.
\end{theorem}

The IDC we will use is specified in Algorithm~\ref{alg:MW}.  The IDC will use approximations $f^{(t)}_{\db}$ in the form of low-degree polynomials of low-weight, and thus we need to specify how to represent such a function.  Specifically, we will represent a polynomial $p$ as a vector $\overline{p}$ of length $2\binom{m}{\leq t} + 1$ with only non-negative entries.  For each coefficient $S \in \binom{[m]}{\leq t}$, the vector will have two components $\overline{p}_{S}, \overline{p}_{\neg S}$.  (Recall that $\binom{[m]}{\leq t} := \set{S \subseteq [m] \mid |S| \leq t}$.)  Intuitively these two entries represent the positive part and negative part of the coefficient $c_{S}$ of $p$.  There will also be an additional entry $\overline{p}_{0}$ that is used to ensure that the $L_1$-norm of the vector is exactly $1$.  Given a polynomial $p \in \cP_{t, W}$ with coefficients $(c_{S})$, we can construct this vector by setting $$\overline{p}_{S} = \frac{\max\{0, c_S\}}{W} \qquad \overline{p}_{\neg S} = \frac{\max\{0, -c_S\}}{W}$$ and choosing $\overline{p}_0$ so that $\|\overline{p}\|_1 = 1$.  Observe that $\overline{p}_0$ can always be set appropriately since the weight of $p$ is at most $W$.

Similarly, we want to associate queries with vectors so that we can replace the evaluation of the polynomial $p$ on a query $y$ with the inner product $W \langle \overline{p}, \overline{y} \rangle$.  We can do so by defining the vector $\overline{y}$ of length $2\binom{m}{\leq t} + 1$ such tha $\overline{y}_{0} = 0$, $\overline{y}_{S} = \prod_{i \in S} y_i$ and $\overline{y}_{\neg S} = - \prod_{i \in S} y_i$.

\begin{fact}
For every $m$-variate polynomial $p$ of degree at most $t$ and weight at most $W$, and every query $y \in \{-1,1\}^m$,  $W\langle \overline{p}, \overline{y} \rangle = p(y)$.  
\end{fact}

\begin{algorithm}
\caption{The Multiplicative Weights Algorithm for Low-Weight Polynomials.}
\label{alg:MW}
$\update_\acc^{MW}(\overline{p}^{(t)},\querystep{t}, a^{(t)}$):
\begin{algorithmic}
\STATE \textbf{Let} $\eta \leftarrow \acc/4W$. 
\STATE{\textbf{If:} $\dbstep{t} = \emptyset$ \textbf{then:} output $$\overline{p}^{(t)} = \frac{1}{2\binom{m}{\leq t} + 1}(1,\dots,1)$$ (representing the constant $0$ polynomial).}
\STATE{\textbf{Else if:} $a^{(t)} <W \langle \overline{p}^{(t)}, \overline{y}^{(t)} \rangle$}
    \INDSTATE[1] \textbf{Let} $\overline{r}^{(t)} = \overline{y}^{(t)}$
\STATE{\textbf{Else:}}
    \INDSTATE[1] \textbf{Let} $\overline{r}^{(t)} = -\overline{y}^{(t)}$
\STATE \textbf{Update:} For all $I \in \set{0, \binom{[m]}{\leq t}, \neg \binom{[m]}{\leq t}}$ let 
$$\overline{p}^{(t+1)}_I \leftarrow \exp(-\eta \overline{r}^{(t)}_{I}) \cdot \overline{p}^{(t)}_I$$ 
$$\overline{p}^{(t+1)} \leftarrow \frac{\overline{p}^{(t+1)}}{\|\overline{p}^{(t+1)}\|_1}$$
\STATE \textbf{Output} $\overline{p}^{(t+1)}$.
\end{algorithmic}
\end{algorithm}

We summarize the properties of the multiplicative weights algorithm in the following theorem:
\begin{theorem}\label{thm:algo-props}
For any $\alpha > 0$, and any family of linear queries $\cQ$ if $\cP_{t,W}$ $(\alpha/4)$-approximates the restriction $\cF|_{Y}$ then Algorithm 1 is an iterative database construction for $\cQ$ with mistake bound $\maxupdates$ for 
$$
\maxupdates = \maxupdates(W, m, t, \alpha) = \frac{16 W^2 \log \left(2\binom{m}{\leq t} + 1\right)}{\alpha^2}.
$$
Moreover, $\update$ runs in time $\poly\big(\binom{m}{\leq t}, \log W, \log(1/\alpha)\big)$.
\end{theorem}

\begin{proof}
Let $\db \in \dbs$ be any database.  For every round $t$ in which $\update$ makes a mistake, we consider the tuple $(\overline{p}^{(t)}, y^{(t)}, \noisyrealstep{t})$ representing the information used to update the approximation in round $t$.  In order to bound the number of mistakes, it will be sufficient to show that after $\maxupdates \leq  16W^2 \log (2\binom{m}{\leq t} + 1) /\acc^2$, the vector $\overline{p}^{(t)}$ is such that
$$
\forall y \in Y_{\cQ}, \; |W \langle \overline{p}^{(t)} , \overline{y} \rangle - f_D(y)| \leq \alpha.
$$
That is, after making $\maxupdates$ mistakes $\overline{p}^{(t)}$ represents a polynomial that approximates $f_{D}$ on every query, and thus there can be no more than $\maxupdates$ makes.

First, we note that there always exists a polynomial $p_D \in \cP_{t, W}$ such that
\begin{equation} \label{eq:priv1}
\forall y \in Y_{\cQ}, \; |p_D(y) - f_D(y)| \leq \frac{\alpha}{4}.
\end{equation}
The assumption of our theorem is that for every $x^{(i)} \in D$, there exists $p_{x^{(i)}} \in \cP_{t,W}$ such that
$$
\forall y \in Y_{\cQ}, \; |p_{x^{(i)}}(y) - f_{x^{(i)}}(y)| \leq \frac{\alpha}{4}.
$$
Thus, since $f_{D} = \frac{1}{n} \sum_{i = 1}^{n} f_{x^{(i)}}$, the polynomial $p_D = \frac{1}{n} \sum_{i=1}^{n} p_{x^{(i)}}$ will satisfy~\eqref{eq:priv1}.  Note that $p_D \in \cP_{t, W}$, thus if we represent $p_D$ as a vector,
$$
\forall y \in Y_{\cQ}, \; |W \langle \overline{p}_D, \overline{y} \rangle - f_D(y)| \leq \frac{\alpha}{4}.
$$

Given the existence of $\overline{p}_D$, we will define a potential function capturing how far $\overline{p}^{(t)}$ is from $\overline{p}_D$.  Specifically, we define
$$
\Psi_t := KL(\overline{p}_D||\overline{p}^{(t)}) = \sum_{I \in \set{0, \binom{[m]}{\leq t}, \neg \binom{[m]}{\leq t}}} \overline{p}_{D,I} \log\left(\frac{\overline{p}_{D,I}}{\overline{p}^{(t)}_{I}}\right)
$$
to be the KL divergence between $\overline{p}_D$ and the current approximation $\overline{p}^{(t)}$.  Note that the sum iterates over all $2\binom{m}{\leq t} + 1$ indices in $\overline{p}$. We have the following fact about KL divergence.

\begin{fact}
For all $t$: $\Psi_t \geq 0$, and $\Psi_0 \leq \log \big(2\binom{m}{\leq t} + 1\big)$.
\end{fact}

We will argue that after each mistake the potential drops by at least $\acc^2/16W^2$.  Note that the potential only changes in rounds where a mistake was made.  Because the potential begins at $\log \big(2\binom{m}{\leq t} + 1\big)$, and must always be non-negative, we know that there can be at most $\maxupdates(\acc) \leq 16W^2 \log \big(2\binom{m}{\leq t} + 1 \big)/\acc^2$ mistakes before the algorithm outputs a (vector representation of) a polynomial that approximates $f_{D}$ on $Y_{\cQ}$.

The following lemma is standard in the analysis of multiplicative-weights based algorithms. 
\begin{lemma}
\label{lem:MWPotentialDrop}
$$\Psi_{t}-\Psi_{t+1} \geq \eta\left( \langle \overline{p}^{(t)}, r^{(t)}\rangle -  \langle \overline{p}_D, r^{(t)}\rangle \right) - \eta^2$$
\end{lemma}
\begin{proof}
\begin{eqnarray*}
\Psi_{t}-\Psi_{t+1} &=& \sum_{I \in \set{0, \binom{[m]}{\leq t}, \neg \binom{[m]}{\leq t}}}  \overline{p}_{D,I} \log\left(\frac{\overline{p}^{(t+1)}_I}{\overline{p}^{(t)}_I}\right) \\
&=& -\eta \langle \overline{p}_{D}, r^{(t)} \rangle -\log\left( \sum_{I \in \set{0, \binom{[m]}{\leq t}, \neg \binom{[m]}{\leq t}}} \exp(-\eta r^{(t)}_{I})\overline{p}^{(t)}_I\right) \\
&\geq& -\eta \langle \overline{p}_{D}, r^{(t)} \rangle - \log\left( \sum_{I \in \set{0, \binom{[m]}{\leq t}, \neg \binom{[m]}{\leq t}}} \overline{p}^{(t)}_I(1+\eta^2-\eta r^{(t)}_{I})\right) \\
&\geq& \eta\left( \langle \overline{p}^{(t)}, r^{(t)} \rangle   -  \langle \overline{p}_D, r^{(t)}\rangle \right) - \eta^2
\end{eqnarray*}
\end{proof}
The rest of the proof now follows easily. By the conditions of an iterative database construction algorithm, $|\noisyrealstep{t} - f_D(y^{(t)})| \leq \acc/2$. Hence, for each $t$ such that $|W \langle \overline{p}^{(t)}, y^{(t)} \rangle - f_{D}(y^{(t)})| \geq \acc$, we also have that $W \langle \overline{p}^{(t)}, y^{(t)} \rangle > f_{D}(y^{(t)})$ if and only if $W \langle \overline{p}^{(t)}, y^{(t)} \rangle > \noisyrealstep{t}$.

In particular, if $r^{(t)} = y^{(t)}$, then $W \langle \overline{p}^{(t)} , y^{(t)} \rangle - W \langle \overline{p}_D , y^{(t)} \rangle \geq \acc/2$.  Similarly, if $r^{(t)} = -\querystep{t}$, then $W \langle \overline{p}_D, y^{(t)} \rangle - W \langle \overline{p}^{(t)},  y^{(t)} \rangle \geq \acc$. Here we have utilized the fact that $|p_D(y) - f_D(y)| \leq \alpha/4$.  Therefore, by Lemma \ref{lem:MWPotentialDrop} and the fact that $\eta = \acc/4W$:
$$\Psi_{t}-\Psi_{t+1} \geq \frac{\acc}{4W}\left(\langle \overline{p}^{(t)}, r^{(t)} \rangle -  \langle \overline{p}_D , r^{(t)} \rangle \right) - \frac{\acc^2}{16W^2} \geq \frac{\acc}{4W}\left(\frac{\acc}{2W} \right) - \frac{\acc^2}{16W^2} = \frac{\acc^2}{16W^2}$$
\end{proof}

Theorem \ref{thm:poly-approx-gives-DP-algorithm} follows immediately from Theorems \ref{thm:winnow} and \ref{thm:algo-props}.

\section{Upper Bounds}
\label{sec:upperbound}
Fact \ref{fact:disj2} and Theorem \ref{thm:poly-approx-gives-DP-algorithm} show that in order to develop a differentially private mechanism that can release all $k$-way marginals of a database,
it is sufficient to construct a low-weight polynomial that approximates the OR$_d$, on all Boolean inputs of Hamming weight at most $k$. 
This is the purpose to which we now turn.


The $\text{OR}_d$ function is easily seen to have an exact polynomial representation of constant weight and degree $d$ (see Fact \ref{fact:and} below); however, an approximation with smaller degree may be achieved at the expense of larger weight. The best known weight-degree tradeoff, implicit in the work of Servedio \etal\ \cite{ThalerCOLT12}, can be stated as follows: there exists a polynomial $p$ of degree $t$ and weight $(d\log{(1/\gamma)/t})^{(d(\log{1/\gamma})^2/t)}$ that $\gamma$-approximates the function $\text{OR}_d$  on all Boolean inputs, for every $t$ larger than $\sqrt{d}\log{(1/\gamma)}$. Setting the degree $t$ to be $O(d/\log^{0.99} d)$ yields a polynomial of weight at most $d^{0.01}$ that approximates the $\text{OR}_d$ function on all Boolean inputs to any desired constant accuracy. On the other hand, Lemma 8 of \cite{ThalerCOLT12} can be shown to imply that any polynomial of weight $W$ that $1/3$-approximates the $\text{OR}_d$ function requires degree $\Omega(d/\log W)$, essentially matching
the $O(d/\log^{.99} d)$ upper bound of Servedio \etal\ when $W=d^{\Omega(1)}$. 

However, in order to privately release $k$-way marginals, we have shown that  it suffices to construct polynomials that are accurate {\em only} on inputs of low Hamming weight. In this section, we give a construction that achieves significantly improved weight degree trade-offs in this setting. In the next section, we demonstrate the tightness of our construction by proving matching lower bounds.

We construct our approximations by decomposing the $d$-variate OR function into an OR of OR's, which is the same approach taken by Servedio \etal\ \cite{ThalerCOLT12}. Here, the outer $\text{OR}$ has fan-in $m$ and the inner $\text{OR}$ has fan-in $d/m$, where the subsequent analysis will determine the appropriate choice of $m$. In order to obtain an approximation that is accurate on all Boolean inputs, Servedio \etal\ approximate the outer OR using a transformation of the Chebyshev polynomials of degree $\sqrt{m}$, and compute each of the inner OR's exactly.  

For $k \ll \log^2 d$, we are able to substantially reduce the degree of the approximating polynomial, relative to the construction of Servedio \etal, by leveraging the fact that we are interested in approximations that are accurate only on inputs of Hamming weight at most $k$. Specifically, we are able to approximate the outer OR function using a polynomial of degree only $\sqrt{k}$ rather than $\sqrt{m}$, and argue 
that the weight of the resulting polynomial is still bounded by a polynomial in $d$.

We now proceed to prove the main lemmas. For the sake of intuition, we begin with weight-degree tradeoffs in the simpler setting in which we are concerned with approximating the $\text{OR}_d$ function over all Boolean inputs. The following lemma, proved below for completeness, is implicit in the work of \cite{ThalerCOLT12}.
\begin{lemma} \label{lem:approx-all}  For every $\gamma>0$ and $m\in [d]$, there is a polynomial of degree $t=O((d/\sqrt{m}) \log (1/\gamma) )$ and weight $W =m^{O(\sqrt{m}\log (1/\gamma))}$ that $\gamma$-approximates the $\text{OR}_d$ function.
\end{lemma}

Our main contribution in this section is the following lemma that gives an improved polynomial approximation to the $\text{OR}_d$ function restricted to inputs of low Hamming weight. 

\begin{lemma}\label{lem:approx-k}
For every $\gamma>0$, $k<d$ and $m\in [d]\setminus [k],$ there is a polynomial of degree $t=O(d \sqrt{k}\log (1/\gamma)/m)$ and weight $W=m^{O(\sqrt{k}\log (1/\gamma))}$ that $\gamma$-approximates the $\text{OR}_d$ function restricted to inputs of Hamming weight at most $k$. 
\end{lemma}
For any constant $\gamma$, one may take $m = d^{O(1/\sqrt{k})}$ in the lemma (here the choice of constant depends on the constants in Fact \ref{fact:chebyshev}) and obtain a polynomial of degree $ d^{1-\Omega(1/\sqrt{k})}$ and weight $d^{0.01}$.  



Our constructions use the following basic facts.   
\begin{fact}\label{fact:and}
The real polynomial $p_d:\{-1,1\}^d \to \R$
$$p_d(x) = 2\left(\sum_{S\subseteq [d]} 2^{-d}\prod_{i\in S} x_i \right)-1 = 2\prod_{i=1}^d\left(\frac{1+x_i}{2}\right)-1$$
computes $\text{OR}_d(x)$ and has weight $w(p_d)\le 3$. 
\end{fact}

\begin{fact}\label{fact:chebyshev}[see e.g.,~\cite{ThalerUlVa12}]
For every $k \in \N$ and $\gamma > 0$, there exists a univariate real polynomial $p = \sum_{i=0}^{t_k} c_i x^i$ of degree $t_k$ such that 
\begin{enumerate}
\item $t_k = O(\sqrt{k}\log (1/\gamma)),$
\item for every $i\in [t_k]$, $|c_i|\leq 2^{O(\sqrt{k}\log(1/\gamma))}$,
\item $p(0)=0$, and
\item for every $x\in [2k]$, $|p(x)-1|\leq \gamma/2$.
\end{enumerate}
\end{fact}


\begin{proof}[Proof of Lemma \ref{lem:approx-all}]
We can compute $\text{OR}_d(y)$ as a disjunction of disjunctions by partitioning the inputs  $y_1,\dots,y_d$ into blocks of size $d/m$ and computing:
$$\text{OR}_{m} (\text{OR}_{d/m}(y_1,\dots,y_{d/m}),\ldots, \text{OR}_{d/m}(y_{d-d/m+1},\ldots,y_{d})).$$
In order to approximately compute $\text{OR}_d(y)$, we compute the inner disjunctions exactly using the polynomial $p_{d/m}$ given in Fact \ref{fact:and} and approximate the outer disjunction using the polynomial from Fact \ref{fact:chebyshev}.  Let 
$$Z (y)= p_{d/m}(y_1,\dots,y_{d/m})+ \cdots + p_{d/m}(y_{d-d/m+1},\ldots,y_{d}).$$
Setting $k=m$ in Fact \ref{fact:chebyshev}, let $q_m$ be the resulting polynomial of degree $O(\sqrt{m}\log (1/\gamma))$ and weight $O(m^{\sqrt{m}\log{(1/\gamma)}})$. Our final polynomial is 
$$1-2q_m(m-Z(y)).$$
Note that $m-Z(y)$ takes values in  $\{0,\ldots,m\}$ and is $0$ exactly when all inputs $y_1,\dots,y_d$ are FALSE.  
It follows that our final polynomial indeed approximates $\text{OR}_d$ to additive error $\gamma$ on all Boolean inputs.

We bound the degree and weight of this polynomial in $y$. By Fact \ref{fact:and}, the inner disjunctions are computed exactly using degree $d/m$ and weight at most $3$. Hence, the total degree is $O(\sqrt{m}\log (1/\gamma) \cdot d/m)$. To bound the weight, we observe that the outer polynomial $q_m(\cdot)$ has at most $T=m^{O(\sqrt{m}\log (1/\gamma))}$ terms where each one has degree at most $D_{\mathrm{outer}}=O(\sqrt{m}\log (1/\gamma))$ and coefficients of absolute value at most $c_{\mathrm{outer}}=2^{O(\sqrt{m}\log{(1/\gamma)})}$. Expanding the polynomials for $Z(y)$, the weight of each term incurs a multiplicative factor 
of $c_{\mathrm{inner}}\le 3^{D_{\mathrm{outer}}}=3^{O(\sqrt{m}\log 1/\gamma)}$ so the total weight is at most $c_{\mathrm{inner}}\cdot c_{\mathrm{outer}}\cdot T=m^{O(\sqrt{m}\log 1/\gamma)}.$  


\end{proof}

\begin{proof}[Proof of Lemma \ref{lem:approx-k}]
Again we partition the inputs $y_1,\dots,y_d$ into blocks of size $d/m$ and view the disjunction as: 
$$\text{OR}_{m} (\text{OR}_{d/m}(y_1,\dots,y_{d/m}),\ldots, \text{OR}_{d/m}(y_{d-d/m+1},\ldots,y_{d})).$$
Once again, we compute the inner disjunctions exactly using the polynomial from Fact \ref{fact:and}.
Let $$Z(y) = p_{d/m}(y_1,\dots,y_{d/m}) + \cdots + p_{d/m}(y_{d-d/m+1},\ldots,y_{d}).$$
If the input $y$ has Hamming weight at most $k$, then $Z(y)$ also takes values in $\{m,\ldots, m-2k\}$.  Thus, we may approximate the outer disjunction using a polynomial of degree $O(\sqrt{k}\log (1/\gamma))$ from Fact \ref{fact:chebyshev}. 
Our final polynomial is:
$$1-2q_k(m - Z).$$
The bound on degree and weight may be obtained as in the previous lemma.  
\end{proof}

\subsection{Proof of upper bound theorems}
\label{sec:main-theorem-proof}
We first present the proof of Theorem \ref{thm:poly-construction}.
\begin{proof}[Proof of Theorem \ref{thm:poly-construction}]
Taking $m = O\big((\log d / \log \log d)^2\big)$ in Lemma~\ref{lem:approx-all} and $m = d^{O(1/\sqrt{k})}$ in Lemma~\ref{lem:approx-k}, it follows that for some constant $C>0$, the $d$-variate disjunction restricted to $H_{d,k}$ is $(1/400)$-approximated by a $d$-variate real polynomial of degree $t$ and weight $W$ where
\[
t=\min\left\{d^{1-\frac{1}{C\sqrt{k}}},\frac{d}{\log^{0.995}{d}}\right\} \text{ and } W=d^{0.01}.
\]


\end{proof}

\begin{proof}[Proof of Theorem \ref{thm:main1}]
By Theorem \ref{thm:poly-construction}, we have a polynomial $p$ that $(1/400)$-approximates the function $OR_d|_{H_{d,k}}$. Moreover, $p$ has weight $W<=d^{0.01}$ and degree 
\[
t\le \min\left\{d^{1-\frac{1}{C\sqrt{k}}},\frac{d}{\log^{0.995}{d}}\right\}.
\]
Thus, by Fact \ref{fact:disj2}, we have a family of polynomials $\cP_{t,W}(d)$ that $(1/400)$-approximates the function family $\cF_{\cQ}|_{H_{d,k}}$. We have the ingredients needed to apply Theorem \ref{thm:poly-approx-gives-DP-algorithm}. Taking $\alpha=1/100$, $\beta=1/100$ gives the conclusion.
\end{proof}

Remark 1 in the Introduction follows from using a slightly different choice of $m$ in Lemma~\ref{lem:approx-all}, namely $m = O( \log^2 d / \log^3 \log d)$.

To obtain the summary of the database promised in Remark 2, we request an answer to each of the $k$-way marginal queries $B(1/400)$ times.  Doing so, will ensure that we obtain a maximal database update sequence, and it was argued in Section~\ref{sec:sans} that the polynomial resulting from any maximal database update sequence accurately answers every $k$-way marginal query.  Finally, we obtain a compact summary by randomly choosing $\tilde{O}(kd^{0.01})$ samples from the normalized coefficient vector of this polynomial to obtain a new sparse polynomial that accurately answers every $k$-way marginal query (see e.g.~\cite{BruckSm92}).  Our compact summary is this final sparse polynomial.

\section{Lower Bounds}
\label{sec:lowerbounds}
In this section, we address the general problem of approximating a block-composed function $G=F(\ldots,f(.),\ldots)$, where $F\from \pmo^k \to \pmo$, $f\from Y\to \pmo$, $Y\subseteq \R^{d/k}$ over inputs restricted to a set $\cY \subseteq Y^k$ using low-weight polynomials. We give a lower bound on the minimum degree of such polynomials. In our main application, $G$ will be $\text{OR}_d$, and $\cY$ will be the set of all $d$-dimensional Boolean
vectors of Hamming weight at most $k$.

Our proof technique is inspired by the composition theorem lower bounds of \cite[Theorem 3.1]{sherstov-intersecths}, where it is shown that the $\gamma$-approximate degree of the composed function $G$ is at least the product of the $\gamma$-approximate degree of the outer function and the PTF degree of the inner function. Our main contribution is a generalization of such a composition theorem along two directions: (1) we show degree lower bounds that take into account the L1-norm of the coefficient vector of the approximating polynomial, and
(2) our lower bounds hold even when we require the approximation to be accurate only on inputs of low Hamming weight, while prior work only considered approximations that are accurate on all Boolean inputs.

Our main theorem is stated below. In  parsing the statement of the theorem, it may be helpful to think of $G=\text{OR}_d$,  $\cY=H_{d,k}$, the set of all $d$-dimensional Boolean
vectors of Hamming weight at most $k$,
 $f = \text{OR}_{d/k}$, $F = \text{OR}_{k}$, $Y=\{-1, 1\}^{d/k}$, and $H = H_{d/k,1}$. This will be the setting of interest in our main application of the theorem.
\begin{theorem}\label{composition-thm:weight-degree-lowerbound}
Let $Y\subset \R^{d/k}$ be a finite set and $\gamma>0$. Given $f:Y\rightarrow \{-1,1\}$ and $F:\{-1,1\}^k\rightarrow \{-1,1\}$ such that $\text{deg}_{2\gamma}(F)=D$, let $G:Y^k\rightarrow \{-1,1\}$ denote the composed function defined by $G(Y_1,\ldots,Y_k)=F(f(Y_1),\ldots,f(Y_k))$. Let $\cY\subseteq Y^k$. Suppose there exists $H\subseteq Y$ such that for every $(Y_1,\ldots,Y_k)\in Y^k\setminus \cY$ there exists $i\in [k]$ such that $Y_i\in Y\setminus H$. Then, for every $t\in \Z_+$,
\[
\text{deg}_{(\gamma,W)}(G|_{\cY})\ge \frac{1}{2}t D\ \text{for every $W\le \gamma 2^{-k}w^*(f|_H,t)^{\frac{D}{2}}$}.
\]
\end{theorem}

We will derive the following corollary from Theorem \ref{composition-thm:weight-degree-lowerbound}. Theorem \ref{thm:main2} follows immediately from Corollary \ref{corr:weight-degree-lowerbound} by considering any $k=o(\log d)$.
\begin{corollary}\label{corr:weight-degree-lowerbound}
Let $k\in [d]$. Then, there exists a universal constant $C>0$ such that 
\[
\text{deg}_{(1/6,W)}(\text{OR}_d|_{H_{d,k}})=\Omega\left(\frac{d}{\sqrt{k}}\cdot \frac{W^{-\frac{1}{C\sqrt{k}}}}{2^{\sqrt{k}/C}}\right)
\text{ for every } W\le \frac{1}{6.2^k}\left(\frac{d}{k}\right)^{C\sqrt{k}}.
\]

\end{corollary}

\paragraph{Intuition underlying our proof technique.}
Recall that our upper bound in Section \ref{sec:upperbound} worked as follows. We viewed $\text{OR}_d$ as an ``OR of ORs'', and we approximated the outer OR with a polynomial $p$ of degree $\text{deg}_{\text{outer}}$ chosen to be as small as possible, and composed $p$ with a low-weight but high-degree polynomial computing each inner OR. We needed to make sure the weight $W_{\text{inner}}$ of the inner polynomials was very low, because the composition step potentially blows the weight up to roughly $W_{\text{inner}}^{\text{deg}_{\text{outer}}}$. As a result, the inner polynomials had to have very high degree, to keep their weight low. 

Intuitively,  we construct a dual solution to a certain linear program that captures the intuition that any low-weight, low-degree polynomial approximation to $\text{OR}_d$ must look something like our primal solution, composing a low-degree approximation to an ``outer'' OR with low-weight approximations to inner ORs. Moreover, our dual solution formalizes the intuition that the composition step must result in a massive blowup in weight,
from $W_{\text{inner}}$ to roughly $W_{\text{inner}}^{\text{deg}_{\text{outer}}}$. 

In more detail, our dual construction works by writing $\text{OR}_d$ as an OR of ORs, where the outer OR is over $k$ variables, and each inner ORs is over $d/k$ variables. We obtain our dual solution by carefully combining a dual witness $\Gamma$ to the high approximate degree of the outer OR, with a dual witness $\psi$ to the fact that any low-degree polynomial with margin at least $1$ for each inner OR, must have ``large'' weight, even if the polynomial must satisfy the margin constraint only on inputs of Hamming weight 0 or 1. This latter condition, that $\psi$ must witness high non-constant margin weight even if restricted to inputs of Hamming weight 0 or 1, is essential to ensuring that our combined dual witness does not place any ``mass'' on irrelevant inputs, i.e. those of Hamming weight larger than $k$.

\subsection{Duality Theorems}
In the rest of the section, we let $\chi_S(x)=\prod_{i\in S}{x_i}$ for any given set $S\subseteq [d]$. 
The question of existence of a weight $W$ polynomial with small degree that $\gamma$-approximates a given function can be expressed as a feasibility problem for a linear program. Now, in order to show the non-existence of such a polynomial, it is sufficient to show infeasibility of the linear program. By duality, this is equivalent to demonstrating the existence of a solution to the corresponding dual program. We begin by summarizing the duality theorems that will be useful in exhibiting this witness.

\begin{theorem}[Duality Theorem for $(\gamma,W)$-approximate degree]\label{thm:duality-weighted-approx-degree}
Fix $\gamma\ge 0$ and let $f:Y\rightarrow \pmo$ be given for some finite set $Y\subset \R^d$. Then, $\text{deg}_{(\gamma,W)}(f)\ge t+1$ if and only if there exists a function $\Psi:Y\rightarrow \R$ such that
\begin{enumerate}
\item $\sum_{y\in Y} |\Psi(y)|=1$,
\item $\sum_{y\in Y} \Psi(y)f(y)-W\cdot|\sum_{y\in Y}\Psi(y)\chi_S(y)|>\gamma$ for every $S\subseteq [d],|S|\le t$.
\end{enumerate}
\end{theorem}
\begin{proof}
By definition, $\text{deg}_{(\gamma,W)}(f)\le t$ if and only if $\exists (\lambda_S)_{S\subseteq[d], |S|\le t}:$
\begin{align*}
\sum_{S\subseteq [d],|S|\le t}|\lambda_S|&\le W, \text{ and} \\
\left|f(y)-\sum_{S\subseteq [d],|S|\le t} \lambda_S \chi_S(y)\right|&\le \gamma\ \forall\ y\in Y.
\end{align*}
By Farkas' lemma, $\text{deg}_{(\gamma,W)}(f)\le t$ if and only if $\nexists\ \Psi:Y\rightarrow \R$ such that 
\[
\frac{1}{W}\sum_{y\in Y}\left(f(y)\Psi(y)-\gamma|\Psi(y)|\right)> \left |\sum_{y\in Y}\chi_S(y)\Psi(y)\right|\ \forall\ S\subseteq [d], |S|\le t.
\]
\end{proof}

The dual witness that we construct to prove Theorem \ref{composition-thm:weight-degree-lowerbound} is obtained by combining a dual witness for the large non-constant margin weight of the inner function with a dual witness for the large approximate degree for the outer function. The duality conditions for these are given below. The proof of the duality condition for the case of $\gamma$-approximate degree is well-known, so we omit the proof for brevity.

\begin{theorem}[Duality Theorem for $\gamma$-approximate degree]\label{thm:duality-approx-degree} \cite{sherstov-pm, spalek, bunthaler}
Fix $\gamma\ge 0$ and let $f:Y\rightarrow \{-1,1\}$ be given, where $Y\subset \R^d$ is a finite set. Then, $\text{deg}_{\gamma}(f)\ge t+1$ if and only if there exists a function $\Gamma:Y\rightarrow \R$ such that
\begin{enumerate}
\item $\sum_{y\in Y} |\Gamma(y)|=1$,
\item $\sum_{y\in Y}\Gamma(y)p(y)=0$ for every polynomial $p$ of degree at most $t$, and
\item $\sum_{y\in Y}\Gamma(y)f(y)>\gamma$.
\end{enumerate}
\end{theorem}

\begin{theorem}[Duality Theorem for non-constant margin weight]\label{thm:duality-non-constant-margin-weight}
Let $Y\subset \R^d$ be a finite set, let $f:Y\rightarrow \{1,-1\}$ be a given function and $w>0$. The non-constant margin weight $w^*(f,t)\ge w$ if and only if there exists a distribution $\mu:Y\rightarrow [0,1]$ such that 
\begin{enumerate}
\item $\sum_{y\in Y}\mu(y)f(y)=0$
\item $\left|\sum_{y\in Y}\mu(y)f(y)\chi_S(y)\right|\le\frac{1}{w}$ for every $S\subseteq[d]$, $|S|\le t$.
\end{enumerate}
\end{theorem}
\begin{proof}
Let $\cS=\{S\subseteq[d]: |S|\le t\}$, $\overline{\cS}=\cS\setminus \{\emptyset\}$. By definition, $w^*(f,t)$ is expressed by the following linear program:
\begin{align*}
\min\ \sum_{S\in \overline{\cS}}|\lambda_S|&\\
f(y)\sum_{S\in \cS} \lambda_S \chi_S(y)&\ge 1\ \forall\ y\in Y.
\end{align*}
The above linear program can be restated as follows:
\begin{align*}
\min\ \sum_{S\in \overline{\cS}}\alpha_S&\\
\alpha_S+\lambda_S&\ge0\ \forall\ S\in \overline{\cS},\\
\alpha_S-\lambda_S&\ge0\ \forall\ S\in \overline{\cS},\\
f(y)\sum_{S\in \cS} \lambda_S \chi_S(y)&\ge 1\ \forall\ y\in Y, \text{ and}\\
\alpha_S&\ge 0\ \forall\ S\in \overline\cS.\\
\end{align*}
The dual program is expressed below:
\begin{align*}
\max\ \sum_y \mu(y)&\\
u_1(S)+u_2(S)&\le 1\ \forall\ S\in \overline\cS,\\
\sum_{y\in Y}\mu(y)f(y)\chi_S(y)+u_1(S)-u_2(S)&=0\ \forall\ S\in \overline\cS,\\
\sum_{y\in Y}\mu(y)f(y)&=0,\\
\mu(y)\ge 0\ \forall\ y\in Y,\ u_1(S),u_2(S)&\ge0\ \forall\ S\in \overline\cS.
\end{align*}

By standard manipulations, the above dual program is equivalent to
\begin{align*}
\max\ \sum_y \mu(y)&\\
|\sum_{y\in Y} \mu(y)\chi_S(y)f(y)|&\le 1 \ \forall\ S\in \overline \cS\\
\sum_{y\in Y}\mu(y)f(y)&=0,\\
\mu(y)&\ge 0\ \forall\ y\in Y
\end{align*}

Finally, given a distribution $\mu'$ satisfying the hypothesis of the theorem, one can obtain a dual solution $\mu$ to show that $w^*(f,t)\ge w$ by taking $w^{-1}=\max_{S\in \cS}|\sum_{y\in Y}\mu'(y)\chi_S(y)f(y)|$ and setting $\mu(y)=w\mu'(y)\ \forall\ y\in Y$. In the other direction, if $w^*(f,t)\ge w$, then we have a dual solution $\mu$ satisfying the above dual program such that $\sum_{y\in Y}\mu(y)=w^*(f,t)$. By setting $\mu'(y)=\mu(y)/w^*(f,t)\ \forall\ y\in Y$, we obtain the desired distribution.
\end{proof}

\subsection{Proof of Theorem~\ref{composition-thm:weight-degree-lowerbound}}
Our approach to exhibiting a dual witness as per Theorem \ref{thm:duality-weighted-approx-degree} is to build a dual witness by appropriately combining the dual witnesses for the ``hardness'' of the inner and outer functions. Our method of combining the dual witnesses is inspired by the technique of \cite[Theorem 3.7]{sherstov-intersecths}. 

\begin{proof}[Proof of Theorem \ref{composition-thm:weight-degree-lowerbound}]
Let $w=w^*(f|_H,t)$. We will exhibit a dual witness function $\Psi:\cY\rightarrow \R$ corresponding to Theorem \ref{thm:duality-weighted-approx-degree} for the specified choice of degree and weight. For $y\in Y^k$, let $Y_i=(y_{(i-1)(d/k)+1},\ldots,y_{i d/k})$. By Theorem \ref{thm:duality-non-constant-margin-weight}, we know that there exists a distribution $\mu:H\rightarrow \R$ such that
\begin{eqnarray}
\sum_{y\in H}\mu(y)f(y)&=&0,\\
\left|\sum_{y\in H}\mu(y)f(y)\chi_S(y)\right|&\le &\frac{1}{w}\ \forall\ S\subseteq\left[\frac{d}{k}\right], |S|\le t \label{eq:mu-low-degree-parity-correlation}
\end{eqnarray}
We set $\mu(y)=0$ for $y\in Y\setminus H$. 

Since $\text{deg}_{2\gamma}(F)=D$, by Theorem \ref{thm:duality-approx-degree}, we know that there exists a function $\Gamma:\{-1,1\}^{k}\rightarrow \R$ such that
\begin{eqnarray}
\sum_	{x\in \{-1,1\}^k}|\Gamma(x)|&=&1, \label{eq:Gamma-l1-norm}\\
\sum_{x\in \{-1,1\}^k}\Gamma(x)p(x)&=&0 \text{ for every polynomial $p$ of degree at most $D$, and}\label{eq:Gamma-fourier-coefficients}\\
\sum_{x\in \{-1,1\}^k}\Gamma(x)F(x)&>&2\gamma. \label{eq:Gamma-correlation}
\end{eqnarray}

Consider the function $\Psi:Y^k\rightarrow \R$ defined as $\Psi(y)=2^{k}\Gamma(f(Y_1),\ldots,f(Y_k))\prod_{i=1}^k{\mu(Y_i)}$. By the hypothesis of the theorem, we know that if $(Y_1,\ldots,Y_k)\in Y^k\setminus \cY$, then there exists $i\in [k]$ such that $Y_i\in Y\setminus H$ and hence $\mu(Y_i)=0$ and therefore $\Psi(Y_1,\ldots,Y_k)=0$.

\begin{enumerate}
\item 
\begin{align*}
\sum_{y\in \cY}|\Psi(y)|&=\sum_{y\in \cY}2^{k}|\Gamma(f(Y_1),\ldots,f(Y_k))|\prod_{i=1}^k{\mu(Y_i)}\\
&=2^k\E_{y\sim \Phi}(|\Gamma(f(Y_1),\ldots,f(Y_k))|)
\end{align*}
where $y\sim \Phi$ denotes $y$ chosen from the product distribution $\Phi:Y^k\rightarrow [0,1]$ defined by $\Phi(y)=\prod_{i\in [k]}\mu(Y_i)$. Since $\sum_{y\in Y}\mu(y)f(y)=0$, it follows that if $Y_i$ is chosen with probability $\mu(Y_i)$, then $f(Y_i)$ is uniformly distributed in $\pmo$. Consequently, 
\[
\sum_{y\in \cY}|\Psi(y)|=2^k\E_{z\sim_U \pmo^k}(|\Gamma(z_1,\ldots,z_k)|)=1.
\]
The last equality is by using (\ref{eq:Gamma-l1-norm}).

\item By the same reasoning as above, it follows from (\ref{eq:Gamma-correlation}) that 
\[
\sum_{y\in \cY}\Psi(y)G(y)=\sum_{z\in \pmo^k}\Gamma(z)F(z)>2\gamma.
\]

\item Fix a subset $S\subseteq[d]$ of size at most $tD/2$. Let $S_i=S\cap \{(i-1)(d/k)+1,\ldots,i d/k\}$ for each $i\in [k]$. Consequently, $\chi_S(y)=\prod_{i=1}^k\chi_{S_i}(Y_i)$.

Now using the Fourier coefficients $\hat\Gamma(T)$ of the function $\Gamma$, we can express
\[
\Gamma(z_1,\ldots,z_k)=\sum_{T\subseteq[k]}\hat\Gamma(T)\prod_{i\in T}z_i=\sum_{\substack{T\subseteq [k],\\ |T|\ge D}}\hat\Gamma(T)\prod_{i\in T}z_i
\]
since $\hat \Gamma(T)=0$ if $|T|<D$ by (\ref{eq:Gamma-fourier-coefficients}). Hence,
\begin{align*}
\Psi(y)&=2^k\sum_{\substack{T\subseteq [k],\\ |T|\ge D}} \hat\Gamma(T)\prod_{i\in T}f(Y_i)\mu(Y_i)\cdot \prod_{i\in[k]\setminus T}\mu(Y_i)
\end{align*}
Therefore,
$\sum_{y\in \cY}\Psi(y)\chi_S(y)$
\begin{align*}
&=\sum_{y\in \cY}\Psi(y)\prod_{i\in [k]}\chi_{S_i}(Y_i)\\
&=2^k\sum_{y\in \cY}\left(\sum_{\substack{T\subseteq [k],\\ |T|\ge D}} \hat\Gamma(T)\prod_{i\in T}f(Y_i)\mu(Y_i)\cdot \prod_{i\in[k]\setminus T}\mu(Y_i)\right) \prod_{i\in [k]}\chi_{S_i}(Y_i)\\
&=2^k\sum_{\substack{T\subseteq [k],\\ |T|\ge D}} \hat\Gamma(T) \sum_{y\in \cY} \left(\prod_{i\in T}f(Y_i)\mu(Y_i)\cdot \prod_{i\in[k]\setminus T}\mu(Y_i) \prod_{i\in [k]}\chi_{S_i}(Y_i)\right)\\
&=2^k\sum_{\substack{T\subseteq [k],\\ |T|\ge D}} \hat\Gamma(T) \sum_{Y_1,\ldots,Y_k\in H} \left(\prod_{i\in T}f(Y_i)\mu(Y_i)\chi_{S_i}(Y_i)\cdot \prod_{i\in[k]\setminus T}\mu(Y_i)\chi_{S_i}(Y_i)\right).
\end{align*}

Rearranging, we have $\sum_{y\in \cY}\Psi(y)\chi_S(y)=$
\begin{equation}
2^{k}\sum_{\substack{T\subseteq [k],\\ |T|\ge D}}\hat\Gamma(T)\prod_{i\in T}\left(\sum_{Y_i\in H}f(Y_i)\mu(Y_i)\chi_{S_i}(Y_i)\right)\prod_{i\in[k]\setminus T}\left(\sum_{Y_i\in H}\mu(Y_i)\chi_{S_i}(Y_i)\right).
\label{eq:RHS}
\end{equation}
Now, we will bound each product term in the outer sum by $w^{-D/2}$. We first observe that for every $i\in [k]$,
\begin{align*}
\sum_{x\in H}\mu(x)\chi_{S_i}(x) &\le \sum_{x\in H}\mu(x) = 1.
\end{align*}
If $|S_i|\le t$, by (\ref{eq:mu-low-degree-parity-correlation})
\[
\left|\sum_{x\in H}f(x)\mu(x)\chi_{S_i}(x) \right|\le \frac{1}{w}.
\]
If $|S_i|>t$, then 
\[
\left|\sum_{x\in H}f(x)\mu(x)\chi_{S_i}(x) \right|\le \sum_{x\in H}\mu(x)= 1.
\]

Since $\sum_{i=1}^k|S_i|\le t D/2$, it follows that $|S_i|\le t$ for more than $k- D/2$ indices $i\in [k]$. 
Thus, for each $T\subseteq [k]$ such that $ |T|\ge  D$, there are at least $ D/2$ indices $i\in T$ such that $|S_i|\le t$. Hence, 
\[
\left|\sum_{y\in \cY}\Psi(y)\chi_S(y)\right|
\le 2^k w^{-\frac{D}{2}}\sum_{T\subseteq [k],|T|\ge D}\left| \hat\Gamma(T) \right|\le 2^{k} w^{-\frac{D}{2}}.
\]
Here, the last inequality is because $|\hat\Gamma(T)|\le 2^{-k}$ from (\ref{eq:Gamma-l1-norm}).
\end{enumerate}

From 1, 2 and 3, we have 
\[
\sum_{y\in \cY}\Psi(y)G(y)-W\max_{S\subseteq [d], |S|\le \frac{tD}{2}}\left|\sum_{y\in \cY}\Psi(y)\chi_S(y)\right|>\gamma
\]
if $W\le \gamma 2^{-k}w^{D/2}$.
\end{proof}

We now derive Corollary \ref{corr:weight-degree-lowerbound}. We need the following theorems on the approximate degree and the non-constant margin weight of the $\text{OR}_d$ function.

\begin{theorem}[Approximate degree of $\text{OR}_d$]\cite{paturi}\label{thm:disj-approx-degree}
$\text{deg}_{1/3}(\text{OR}_d)=\Theta(\sqrt{d})$.
\end{theorem}

\begin{lemma}[Non-constant margin weight of $\text{OR}_d$]\label{lem:disj-non-constant-margin-weight}
$w^*(\text{OR}_d|_{H_{d,1}},t)\ge d/t.$
\end{lemma}
\begin{proof}
The function 
\begin{equation*}
\mu(x)=\left\{
\begin{array}{rl}
1/2 &\text{ if } x=(1,\ldots,1),\\
{1}/{2d} &\text{ if } x\in H_{d,1}\setminus \{(1,\ldots,1)\}.
\end{array}
\right.
\end{equation*}
acts as the dual witness in Theorem \ref{thm:duality-non-constant-margin-weight}.
\end{proof}

\begin{proof}[Proof of Corollary \ref{corr:weight-degree-lowerbound}]
We use Theorem \ref{composition-thm:weight-degree-lowerbound} in the following setting. Let $Y=\pmo^{d/k}$, the inner function $f:Y\rightarrow \pmo$ be $\text{OR}_{d/k}$ and the outer function $F:\{-1,1\}^k\rightarrow \{-1,1\}$ be $\text{OR}_k$, $\cY=H_{d,k}$ and $H=H_{d/k,1}$. By a simple counting argument, if $(Y_1,\ldots,Y_k)\in \pmo^d\setminus H_{d,k}$, then there exists $i\in [k]$ such that $Y_i\in \pmo^{d/k}\setminus H_{d/k,1}$. Further, by Theorem \ref{thm:disj-approx-degree}, we know that $\text{deg}_{1/3}(F)=\Theta(\sqrt{k})$ and by Claim \ref{lem:disj-non-constant-margin-weight}, we know that $w^*(f|_H,t)\ge d/kt$. Therefore, by Theorem \ref{composition-thm:weight-degree-lowerbound}, we have that, for every $t\in \Z_+$, 
\[
\text{deg}_{1/6,W}(OR_d|_{H_{d,k}})=\Omega\left(t\sqrt{k}\right) \text{ for every } W\le \frac{1}{6}2^{-k}\left(\frac{d}{kt}\right)^{C\sqrt{k}}.
\]
We obtain the conclusion by taking $t=\lfloor(d/k)(6W2^{k})^{-1/C\sqrt{k}}\rfloor$. Since $W\le (1/6)2^{-k}(d/k)^{C\sqrt{k}}$, it follows that $t\ge 1$ and hence is a valid choice for $t$ in applying Theorem \ref{composition-thm:weight-degree-lowerbound}.
\end{proof}

\paragraph{Comparison to \cite{ThalerCOLT12}.}
As described in the beginning of Section \ref{sec:upperbound}, Lemma 8 of the work of Servedio \etal\ \cite{ThalerCOLT12} can be shown to imply that any polynomial $p$ of weight $W$ that $1/3$-approximates the $\text{OR}_d$ function on all Boolean inputs requires degree $\Omega(d/\log W).$\footnote{More precisely,  \cite[Lemma 8]{ThalerCOLT12} as stated shows that if the coefficients of a univariate polynomial $P$ each have absolute value at most $W$, and  $1/2 \leq \max_{x \in [0, 1]} |P(x)| \leq R$, then $\max_{x \in [0,1]} |P'(x)| = O(\text{deg}(P) \cdot R \cdot  (\log W + \log \text{deg}(P)))$, where $P'(x)$ denotes the derivative of $P$ at $x$, and $\deg(P)$ denotes the degree of $P$. By inspection of the proof, it is easily seen that if the $L_1$-norm of the coefficients of $P$ is bounded by $W$, then the following slightly stronger conclusion holds: $\max_{x \in [0,1]} |P'(x)| = O(\text{deg}(P) \cdot R \cdot \log W)$. When combined with the symmetrization argument of \cite{ThalerCOLT12}, this stronger conclusion implies: any polynomial $p$ of weight $W$ that $1/3$-approximates the $\text{OR}_d$ function on all Boolean inputs requires degree $\Omega(d/\log W)$.}
 The proof in \cite{ThalerCOLT12} relies on a Markov-type inequality that bounds the derivative of a univariate polynomial in terms of its degree and the size of its
coefficients. The proof of this Markov-type inequality is non-constructive and relies on complex analysis. 

Here, we observe that our dual witness construction used to prove Corollary \ref{corr:weight-degree-lowerbound} also yields a general lower bound on the tradeoffs achievable between the weight and degree of the approximating polynomial $p$, even when we require $p$ to be accurate only on inputs of Hamming weight at most $O(\log W)$ (see Theorem \ref{thm:recover-COLT}). The methods of Servedio \etal\ do not yield any non-trivial lower bound on the degree in this setting. We also believe our proof technique is of interest in comparison to the methods of Servedio \etal\ as it is constructive (exhibiting an explicit dual witness for the lower bound) and avoids the use of complex analysis. 

\begin{theorem}\label{thm:recover-COLT} Any polynomial $p$ of weight $W$ that $1/6$-approximates the $\text{OR}_d$ function on all Boolean inputs requires degree $d/2^{O(\sqrt{\log W})}$. \end{theorem}

\begin{proof} 
As $p$ is accurate on the entire Boolen hypercube, it is accurate on inputs of Hamming weight at most $\log W$. The theorem follows by setting $k=\log{W}$ in the statement of Corollary \ref{corr:weight-degree-lowerbound}. 
\end{proof}




\section{Discussion}

We gave a differentially private online algorithm for answering $k$-way marginal queries that runs in time $2^{o(d)}$ per query, and guarantees accurate answers for databases of size $\poly(d,k).$  More precisely, we showed that if there exists a polynomial of degree $t$ and weight $W$ approximating the $d$-variate OR function on Boolean inputs of Hamming weight at most $k$, then a variant of the private multiplicative weights algorithm can answer $k$-way marginal queries in time roughly $\binom{d}{t}$ per query and guarantee accurate answers on databases of size roughly $W\sqrt{d}$.  To this end, we gave a new construction showing the existence of polynomial approximations to the OR function on inputs of low Hamming weight.  Specifically, we showed that polynomials of weight $d^{0.01}$ and degree $d^{1-\Omega(1/\sqrt{k})}$ exist.  

\paragraph{Practical Considerations.}
Our algorithm for answering $k$-way marginals is essentially the same as in \cite{HardtRo10} except for using a different set of base functions (specifically, the set of all low-degree parities), which leads to an efficiency gain.  We note that our algorithm degrades smoothly to the private multiplicative weights algorithm as the degree of the promised polynomial approximation increases, and never gives a worse running time. This behavior suggests that our algorithm may lead to practical improvements even for relatively small values of $d$, for which the asymptotic analysis does not apply.  In such cases one might use an alternative but similar analysis that shows the existence of a polynomial of degree $kd^{1-c/k}$ and weight $d^c$ (for any $0 < c < k$) that \emph{exactly} computes the $d$-variate OR function on inputs of Hamming weight at most $k$.  Such a polynomial may be obtained as in our construction, by breaking the $d$-variate OR function into an OR of ORs, and using a degree $k$ polynomial defined via polynomial interpolation, instead of Chebyshev polynomials, to approximate the outer OR on inputs of Hamming weight at most $k$. This variant does achieve asymptotic
properties not shared by the algorithm of Theorem \ref{thm:main1}, owing to the fact that it uses exact rather than approximate
representations of the database: for any constant $k$, this variant
also runs in time $\exp\left(d^{1-\Omega(1)}\right)$, and achieves worst-case additive error $o(1)$ for sufficiently large databases, i.e., for $n \gtrsim \tilde{O}(d^{c})$.

\paragraph{Relationship with~\cite{NikolovTaZh13,DworkNiTa13}.}
As we mentioned in the introduction, in subsequent work, Dwork et al. ~\cite{DworkNiTa13} show how to privately release marginals in a very different parameter regime from what we consider here.  Although their algorithm is quite different from ours, there are some important similarities.  Their algorithm is based on a recent algorithm of Nikolov et. al. ~\cite{NikolovTaZh13} for answering arbitrary counting queries.  This algorithm proceeds as follows. First, it adds noise $O(|\cQ|^{1/2} / |\db|)$ to the answer to every query to obtain a vector of noisy answers.  The noise is sufficient to ensure differential privacy.  Second, it ``projects'' the vector of noisy answers onto the convex body $K$ consisting of all vectors of answers that are consistent with some real database.  Surprisingly, the projection step will improve the accuracy. In fact, Nikolov et al.~\cite{NikolovTaZh13} show that the projected answers will be accurate even for very small databases. However current best algorithms for projecting on the body $K$ require time $2^{O(d)}$ for an arbitrary set of queries.

Dwork et al.~\cite{DworkNiTa13} show that for the set of $k$-way marginals there is a convex body $L$ such that the projection into $L$ can be computed in time $\poly(\binom{d}{\leq k})$ and $L$ approximates $K$ well enough to achieve accuracy on smaller databases than would be achievable with independent noise (however, these databases still have size $d^{\Omega(k)}$).  Our approximation of $f_{D}$ by low-weight polynomials of degree $t = o(d)$ can be shown to imply a polytope $L'$ with $2^{o(d)}$ vertices---which is sufficient to imply that projection can be computed in time $2^{o(d)}$---that approximates $K$ well enough to achieve accuracy on databases of nearly optimal size (i.e., size roughly $kd^{0.51})$. Thus our approximation-theoretic approach is relevant for understanding the capabilities and limitations of algorithms in the Nikolov et al.~\cite{NikolovTaZh13} framework.

\paragraph{Future Directions.}
Our lower bounds show that our polynomial approximation to the $\text{OR}_d$ function on inputs of Hamming weight $k$ is essentially the best possible; in particular, we cannot hope to substantially improve the running time on $\poly(d,k)$ size databases by giving approximating polynomials with better weight and degree bounds.  
This also rules out several natural candidates that can themselves be computed exactly by a low-weight polynomial of low-degree (e.g., the set of small-width conjunctions).

However, we do not know if it is possible to do better by using different feature spaces (other than the set of all low-degree monomials) to approximate all disjunctions over $d$ variables. There is some additional evidence from prior work that low-degree monomials may be the optimal choice: if the parameter of significance is the \emph{size} of the set of functions used to approximate disjunctions on inputs of Hamming weight at most $k$, then low-degree monomials are indeed optimal \cite{sherstov-pm} (see also Section 5 in the full version of \cite{ThalerUlVa12}). It would be interesting to determine whether this optimality still holds  when we restrict the $L_1$ weight of the linear combinations used in the approximations to be $\poly(d)$.  \\

\noindent {\bf Acknowledgments}. We thank Salil Vadhan for helpful discussions about this work.

\bibliographystyle{alpha}
\bibliography{references}

\newcommand{\etalchar}[1]{$^{#1}$}
\begin{thebibliography}{BDMN05}

\bibitem[BCD{\etalchar{+}}07]{BarakChDwKaMcTa07}
Boaz Barak, Kamalika Chaudhuri, Cynthia Dwork, Satyen Kale, Frank McSherry, and
  Kunal Talwar.
\newblock Privacy, accuracy, and consistency too: a holistic solution to
  contingency table release.
\newblock In {\em PODS}, pages 273--282, 2007.

\bibitem[BDMN05]{BlumDwMcNi05}
Avrim Blum, Cynthia Dwork, Frank McSherry, and Kobbi Nissim.
\newblock Practical privacy: the sulq framework.
\newblock In {\em PODS}, pages 128--138, 2005.

\bibitem[BLR08]{BlumLiRo08}
Avrim Blum, Katrina Ligett, and Aaron Roth.
\newblock A learning theory approach to non-interactive database privacy.
\newblock In {\em STOC}, pages 609--618, 2008.

\bibitem[BS92]{BruckSm92}
Jehoshua Bruck and Roman Smolensky.
\newblock {Polynomial Threshold Functions, $AC^0$ Functions, and Spectral
  Norms}.
\newblock {\em SIAM J. Comput.}, 21(1):33--42, 1992.

\bibitem[BT13]{bunthaler}
Mark Bun and Justin Thaler.
\newblock Dual lower bounds for approximate degree and markov-bernstein
  inequalities.
\newblock In {\em Proceedings of 40th International Colloquium on Automata,
  Languages, and Programming}, pages 303--314, 2013.

\bibitem[BUV13]{BunUlVa13}
Mark Bun, Jonathan Ullman, and Salil Vadhan.
\newblock Fingerprinting codes and the true price of differential privacy.
\newblock {\em Manuscript}, 2013.

\bibitem[CKKL12]{CheraghchiKlKoLe12}
Mahdi Cheraghchi, Adam Klivans, Pravesh Kothari, and Homin~K. Lee.
\newblock Submodular functions are noise stable.
\newblock In {\em Proceedings of the Twenty-Second Annual ACM-SIAM Symposium on
  Discrete Algorithms}, SODA, pages 1586--1592, 2012.

\bibitem[De12]{De12}
Anindya De.
\newblock Lower bounds in differential privacy.
\newblock In {\em TCC}, pages 321--338, 2012.

\bibitem[DMNS06]{DworkMcNiSm06}
Cynthia Dwork, Frank McSherry, Kobbi Nissim, and Adam Smith.
\newblock Calibrating noise to sensitivity in private data analysis.
\newblock In {\em TCC '06}, pages 265--284, 2006.

\bibitem[DN03]{DinurNi03}
Irit Dinur and Kobbi Nissim.
\newblock Revealing information while preserving privacy.
\newblock In {\em PODS}, pages 202--210, 2003.

\bibitem[DN04]{DworkNi04}
Cynthia Dwork and Kobbi Nissim.
\newblock Privacy-preserving datamining on vertically partitioned databases.
\newblock In {\em CRYPTO}, pages 528--544, 2004.

\bibitem[DNR{\etalchar{+}}09]{DworkNaReRoVa09}
Cynthia Dwork, Moni Naor, Omer Reingold, Guy~N. Rothblum, and Salil~P. Vadhan.
\newblock On the complexity of differentially private data release: efficient
  algorithms and hardness results.
\newblock In {\em STOC '09}, pages 381--390, 2009.

\bibitem[DNT13]{DworkNiTa13}
Cynthia Dwork, Aleksandar Nikolov, and Kunal Talwar.
\newblock Efficient algorithms for privately releasing marginals via convex
  relaxations.
\newblock {\em CoRR}, abs/1308.1385, 2013.

\bibitem[DRV10]{DworkRoVa10}
Cynthia Dwork, Guy~N. Rothblum, and Salil~P. Vadhan.
\newblock Boosting and differential privacy.
\newblock In {\em FOCS}, pages 51--60, 2010.

\bibitem[FK13]{FeldmanKo13}
Vitaly Feldman and Pravesh Kothari.
\newblock Learning coverage functions.
\newblock {\em Manuscript}, 2013.

\bibitem[GHRU11]{GuptaHaRoUl11}
Anupam Gupta, Moritz Hardt, Aaron Roth, and Jonathan Ullman.
\newblock Privately releasing conjunctions and the statistical query barrier.
\newblock In {\em STOC '11}, pages 803--812, 2011.

\bibitem[GRU12]{GuptaRoUl12}
Anupam Gupta, Aaron Roth, and Jonathan Ullman.
\newblock Iterative constructions and private data release.
\newblock In {\em TCC}, pages 339--356, 2012.

\bibitem[HLM12]{HardtLiMc12}
Moritz Hardt, Katrina Ligett, and Frank McSherry.
\newblock A simple and practical algorithm for differentially private data
  release.
\newblock {\em NIPS '12}, 2012.

\bibitem[HR10]{HardtRo10}
Moritz Hardt and Guy~N. Rothblum.
\newblock A multiplicative weights mechanism for privacy-preserving data
  analysis.
\newblock In {\em FOCS}, pages 61--70, 2010.

\bibitem[HRS12]{HardtRoSe12}
Moritz Hardt, Guy~N. Rothblum, and Rocco~A. Servedio.
\newblock Private data release via learning thresholds.
\newblock In {\em SODA}, pages 168--187, 2012.

\bibitem[JT12]{JainTh12}
Prateek Jain and Abhradeep Thakurta.
\newblock Mirror descent based database privacy.
\newblock In {\em APPROX-RANDOM}, pages 579--590, 2012.

\bibitem[Kla11]{klauck}
Hartmut Klauck.
\newblock On arthur merlin games in communication complexity.
\newblock In {\em Proceedings of the 26th Annual Conference on Computational
  Complexity (CCC)}, CCC, pages 189--199, 2011.

\bibitem[KRSU10]{KasiviswanathanRuSmUl10}
Shiva~Prasad Kasiviswanathan, Mark Rudelson, Adam Smith, and Jonathan Ullman.
\newblock The price of privately releasing contingency tables and the spectra
  of random matrices with correlated rows.
\newblock In {\em STOC}, pages 775--784, 2010.

\bibitem[KS04]{KlivansSe04}
Adam~R. Klivans and Rocco~A. Servedio.
\newblock Toward attribute efficient learning of decision lists and parities.
\newblock In John Shawe-Taylor and Yoram Singer, editors, {\em COLT}, volume
  3120 of {\em Lecture Notes in Computer Science}, pages 224--238. Springer,
  2004.

\bibitem[LS13]{LS13}
Philip~M. Long and Rocco~A. Servedio.
\newblock Low-weight halfspaces for sparse boolean vectors.
\newblock In {\em ITCS}, pages 21--36, 2013.

\bibitem[NTZ13]{NikolovTaZh13}
Aleksandar Nikolov, Kunal Talwar, and Li~Zhang.
\newblock The geometry of differential privacy: the sparse and approximate
  cases.
\newblock In {\em STOC}, pages 351--360, 2013.

\bibitem[Pat92]{paturi}
R.~Paturi.
\newblock On the degree of polynomials that approximate symmetric boolean
  functions.
\newblock In {\em Proceedings of the twenty-fourth annual ACM Symposium on
  Theory of Computing}, pages 468--474, 1992.

\bibitem[RR10]{RothRo10}
Aaron Roth and Tim Roughgarden.
\newblock Interactive privacy via the median mechanism.
\newblock In {\em STOC '10}, pages 765--774, 2010.

\bibitem[She09]{sherstov-intersecths}
Alexander~A. Sherstov.
\newblock The intersection of two halfspaces has high threshold degree.
\newblock In {\em Proceedings of the 2009 50th Annual IEEE Symposium on
  Foundations of Computer Science}, FOCS '09, pages 343--362, 2009.

\bibitem[She11]{sherstov-pm}
A.~A. Sherstov.
\newblock The pattern matrix method.
\newblock {\em SIAM Journal on Computing}, 40(6):1969--2000, 2011.

\bibitem[She12a]{sherstov-multi}
Alexander~A. Sherstov.
\newblock The multiparty communication complexity of set disjointness.
\newblock In {\em STOC}, pages 525--548, 2012.

\bibitem[She12b]{sherstov-direct}
Alexander~A. Sherstov.
\newblock Strong direct product theorems for quantum communication and query
  complexity.
\newblock {\em SIAM J. Comput.}, 41(5):1122--1165, 2012.

\bibitem[STT12]{ThalerCOLT12}
Rocco Servedio, Li-Yang Tan, and Justin Thaler.
\newblock Attribute-efficient learning and weight-degree tradeoffs for
  polynomial threshold functions.
\newblock In {\em 25th Annual Conference on Computational Learning Theory
  (COLT), JMLR Workshop and Conference Proceeding}, volume~23, pages
  14.1--14.19, 2012.

\bibitem[TUV12]{ThalerUlVa12}
Justin Thaler, Jonathan Ullman, and Salil~P. Vadhan.
\newblock Faster algorithms for privately releasing marginals.
\newblock In {\em ICALP}, pages 810--821, 2012.

\bibitem[Ull13]{Ullman13}
Jonathan Ullman.
\newblock Answering n$^{\mbox{2+o(1)}}$ counting queries with differential
  privacy is hard.
\newblock In {\em STOC}, pages 361--370, 2013.

\bibitem[UV11]{UllmanVa11}
Jonathan Ullman and Salil~P. Vadhan.
\newblock {PCP}s and the hardness of generating private synthetic data.
\newblock In {\em TCC '11}, pages 400--416, 2011.

\bibitem[\v{S}08]{spalek}
Robert \v{S}palek.
\newblock A dual polynomial for {OR}.
\newblock {\em CoRR}, abs/0803.4516, 2008.

\end{thebibliography}

\end{document}